\documentclass[runningheads]{llncs}
\usepackage[utf8]{inputenc}
\usepackage[T1]{fontenc}
\usepackage{amsmath,amsfonts}  
\usepackage{microtype}  
\usepackage{graphicx}  
\usepackage{lipsum}
\usepackage{tcolorbox}
\usepackage{thm-restate}
\usepackage[inline]{enumitem}
\usepackage{parskip}
\usepackage[ruled,vlined,linesnumbered]{algorithm2e}
\usepackage{setspace}
\usepackage{ragged2e}
\usepackage{mathtools}
\usepackage[]{todonotes}
\usepackage{wrapfig}
\usepackage{multicol}
\usepackage{amscd}
\usepackage{nccmath}
\usepackage{tikz-cd}
\usepackage{multicol}
\usepackage{lineno}
\usepackage{graphicx}
\usepackage{hyperref}
\usepackage[capitalise]{cleveref}
\usepackage{multicol}
\usepackage[pdf]{graphviz}
\usepackage{qcircuit}
\usepackage[noend]{algpseudocode}
\usepackage[english]{babel}
\usepackage[autosize]{dot2texi}
\usepackage{tikz}
\usepackage{caption}
\usepackage{subcaption}
\usepackage{stmaryrd}

\usepackage{adjustbox}

\tcbuselibrary{skins}

\newcommand{\cnot}{\text{\fontfamily{cmtt}\selectfont CNOT}}

\newcommand{\swap}{\text{\fontfamily{cmtt}\selectfont SWAP}}

\newcommand{\SuchThat}{~\vert~}

\newcommand{\Distance}{\delta}
\newcommand{\Sphere}{R}
\newcommand{\Path}{\rightsquigarrow}
\newcommand{\Diameter}{\mathsf{diam}}
\newcommand{\BreakingPoint}{n_0}
\newcommand{\NumPermutationCycles}{c}
\newcommand{\Quotient}{q}
\newcommand{\Remainder}{r}
\newcommand{\Essential}{\varepsilon}
\newcommand{\Embedding}{\phi}
\newcommand{\OrbitsOfSymmetriesOnSphere}{\mathcal{D}}
\newcommand{\OrbitsOfSymmetriesWithEssential}{\mathcal{E}}
\newcommand{\OrbitsOfSymmetriesOnSphereWithEssential}{\mathcal{C}}
\newcommand{\InverseTransposeOrbitCount}{\kappa}
\newcommand{\Bijection}{\Phi}

\newcommand{\TVShort}[1]{\llbracket #1\rrbracket}

\newcommand{\Group}{\mathcal{G}}
\newcommand{\Nats}{\mathbb{N}}
\newcommand{\NatsPlus}{\Nats^+}
\newcommand{\GL}{\mathsf{GL}}
\newcommand{\SymmetricGroup}{\mathcal{S}}
\newcommand{\F}{\mathbb{F}}
\newcommand{\TransvectionGenerators}{\Sigma}
\newcommand{\Stabilizer}{\mathsf{stab}}
\newcommand{\Automorphisms}{\mathsf{aut}}
\newcommand{\Isometries}{\mathsf{isom}}
\newcommand{\CyclicGroup}{\mathcal{C}}
\newcommand{\J}{\mathcal{J}}

\newcommand{\PermutationMatrix}{P}
\newcommand{\TransvectionMatrix}{T}
\newcommand{\IdentityMatrix}{I}

\newcommand{\Queue}{\mathcal{Q}}
\newcommand{\Push}{\mathsf{push}}
\newcommand{\Pop}{\mathsf{pop}}
\newcommand{\DistanceMap}{\mathsf{dist}}
\newcommand{\SphereSize}{\mathsf{SphSize}}
\newcommand{\Representative}{\mathsf{Rep}}

\newcommand{\Paragraph}[1]{\smallskip\noindent{\bf #1}}
\newcommand{\SubParagraph}[1]{\smallskip\noindent{\em #1}}
\newcommand\numberthis{\addtocounter{equation}{1}\tag{\theequation}}

\IncMargin{1em}
\SetInd{0.35em}{0.35em}

\SetCommentSty{mycommfont}
\newlist{compactenum}{enumerate}{3} 
\setlist[compactenum]{label=(\arabic*), nosep,leftmargin=*}
\crefname{compactenumi}{Item}{Items}

\newlist{compactitem}{itemize}{3} 
\setlist[compactitem]{label=$\bullet$, nosep,leftmargin=*}

\crefname{conjecture}{Conjecture}{Conjectures}

\begin{document}
\title{On Exact Sizes of Minimal CNOT Circuits}
%
%
\author{
Jens Emil Christensen\inst{1} \and
S\o ren Fuglede J\o rgensen\inst{2}\and \\
Andreas Pavlogiannis\inst{1}
\and
Jaco van de Pol\inst{1}
}
\authorrunning{J. E. Christensen  et al.}
%
\institute{Aarhus University, Aarhus, Denmark \and 
Kvantify ApS, Copenhagen, Denmark
\\
\email{202005655@post.au.dk}\\
\email{sfj@kvantify.dk}\\
\email{\{pavlogiannis,jaco\}@cs.au.dk}
}
\maketitle              

\begin{abstract}
Computing a minimum-size circuit that implements a certain function is a standard optimization task.
We consider circuits of \cnot\ gates, which are fundamental binary gates in reversible and quantum computing.
Algebraically, \cnot\ circuits on $n$ qubits correspond to $\GL(n,2)$, the general linear group over the field of two elements,
and circuit minimization reduces to computing distances in the Cayley graph $G_n$ of $\GL(n,2)$ generated by transvections.
However, the super-exponential size of $\GL(n,2)$ has made its exploration computationally challenging.
\\[1em]
In this paper, we develop a new approach for computing distances in $G_n$, allowing us to synthesize minimum circuits that were previously beyond reach (e.g., we can synthesize optimally all circuits over $n=7$ qubits).
Towards this, we establish two theoretical results that may be of independent interest.
First, we give a complete characterization of all isometries in $G_n$ in terms of
(i)~permuting qubits and
(ii)~swapping the arguments of all \cnot\ gates.
Second, for any fixed $d$, we establish polynomials in $n$ of degree $2d$ that characterize the size of spheres in $G_n$ at distance $d$ from the identity, as long as $n\geq 2d$.
With these tools, we revisit an open question of [Bataille, 2020] regarding the smallest number $n_0$ for which the diameter of $G_{n_0}$ exceeds $3(n_0-1)$.
It was previously shown that $6\leq n_0 \leq 30$, a gap that we tighten considerably to $8\leq n_0 \leq 20$.
We also confirm a conjecture that long cycle permutations lie at distance $3(n-1)$, for all $n\leq 8$, extending the previous bound of $n\leq 5$.

\end{abstract}
\section{Introduction}\label{sec:intro}

\cnot\ circuits, also known as linear reversible circuits, are fundamental in reversible and quantum computing.
A \cnot\ gate operates on two inputs, a control bit $c$ and a data bit $d$, having the effect $\cnot(c,d)=(c,c\oplus d)$, i.e. $d$ is negated if $c$ is on.
In quantum computing, the effect of a \cnot\ gate is extended to linear combinations of qubits, and is crucial to create entanglement, as in common gate sets of theoretical and practical interest, \cnot s are the only non-unary gates~\cite{Dawson2006}.
In physical realizations, executing binary \cnot\ gates is a major cause of noise~\cite{Linke2017}.
As such, the problem of reducing the \cnot\ count of a circuit, or finding an equivalent circuit with minimal \cnot-count, is an active research topic~\cite{Schneider2022ASE,DBLP:conf/rc/MeuliSM18,9914638,Amy2018}.

The end-to-end function of a \cnot-circuit $C$ on $n$ qubits is captured by a $n\times n$ parity matrix $M$. 
Starting from the identity matrix $\IdentityMatrix_n$, each \cnot($c,d$) adds the $c$-th row to the $d$-th row. 
Then, an optimal circuit $C'$ for $M$ corresponds to the minimal number of row additions required to obtain $M$ from $\IdentityMatrix_n$ (see \cref{fig:example}).
Due to their theoretical elegance and practical importance, the optimization specifically of \cnot\ circuits has received special attention, e.g., via SAT solvers~\cite{DBLP:conf/ecai/ShaikP24},
and heuristics~\cite{Patel2008,10.1145/3474226,de2021reducing,DBLP:conf/rc/BrugiereBVMA20,Kazuo2002} (lacking optimality guarantees in general).

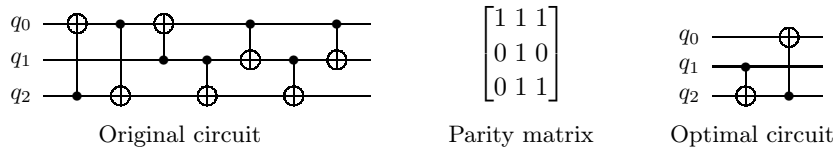
\begin{figure}[t]
\begin{subfigure}[b]{0.45\textwidth}
\scalebox{0.9}{
\hspace{0.8cm}
\Qcircuit @C=1em @R=.7em {
             \lstick{q_0} & \targ    & \ctrl{2} & \targ     & \qw         & \ctrl{1} & \qw      & \ctrl{1} & \qw \\
             \lstick{q_1 }     & \qw      & \qw      &\ctrl{-1} & \ctrl{1}     & \targ    & \ctrl{1} & \targ    & \qw \\
             \lstick{q_2} & \ctrl{-2} & \targ &\qw       & \targ     & \qw      & \targ    & \qw      & \qw \\
        }}
    \caption*{Original circuit}
\end{subfigure}
\hfill
\begin{subfigure}[b]{0.2\textwidth}
\centering
$
\begin{bmatrix}
    1 & 1 & 1 \\
    0 & 1 & 0 \\
    0 & 1 & 1
\end{bmatrix}
$    \caption*{Parity matrix}
    \end{subfigure}
\hfill
\begin{subfigure}[b]{0.21\textwidth}
\scalebox{0.9}{
\hspace{0.6cm}
\centering
\Qcircuit @C=1em @R=.7em {
     \lstick{q_0}  & \qw     & \targ  & \qw \\
     \lstick{q_1 } & \ctrl{1}   & \qw  & \qw \\
     \lstick{q_2}  & \targ     & \ctrl{-2}  & \qw \\
}}
\caption*{Optimal circuit}
    \end{subfigure}
\hfill
\medskip
    \caption{CNOT circuit optimization using the parity matrix.}
    \label{fig:example}
\end{figure}

Algebraically, the \cnot\ operators can be viewed as transvections, a set of generators of the linear group of $n\times n$-matrices over ${\mathbb F}_2$, i.e., $\GL(n,2)$. 
The size of an optimal circuit for a matrix $M$ corresponds to the distance of $M$ from $\IdentityMatrix_n$ in the corresponding Cayley graph $G_n$.
The diameter of $G_n$ corresponds to the size of the largest optimal circuit on $n$ qubits and grows as
$\Theta(n^2/\log{n})$~\cite{Patel2008}.
The relation of distances to optimal circuits has spurred interest in exploring Cayley graphs for a small number of qubits, and computing their diameter~\cite{Bataille2020,Bataille2022}.

\Paragraph{Contributions.}
Our main contributions are as follows.

\begin{compactenum}
\item We develop Isometry BFS as a general, breadth-first exploration of the Cayley graph $G$ of an arbitrary group $\Group$ generated by a set of generators, based on general isometries $\J$.
This allows us to store a single representative from each orbit of $\J$, reducing the memory footprint of the exploration to roughly $O(\lvert\Group/\J\rvert)$.
The lower memory also allows one to store $G$ as a database for looking up the shortest products generating an element of $\Group$\footnote{Techniques resembling ours were developed recently specifically for the Clifford group in quantum computing, covering the Cayley graph over $6$ qubits~\cite{Bravy-6qubit}.}.
For our \cnot\ case, Isometry BFS enables us to synthesize optimally all \cnot\ circuits over $n=7$ qubits, extending the previous bound of $n=5$~\cite{Bataille2022}.
\item We revisit an open question of~\cite{Bataille2022} regarding the smallest number $n_0$ for which the diameter of $G_{n_0}$ exceeds $3(n_0-1)$.
It was previously shown that  $6\leq n_0 \leq 30$.
We tighten this gap considerably, to $8\leq n_0 \leq 20$ effectively halving the previous one.
\item We also revisit a conjecture that permutation matrices of $\GL(n,2)$ whose cycle types consist of $p$ cycles lie at distance $3(n-p)$ in $G_n$~\cite{Bataille2020}.
We confirm the conjecture for all $n\leq 8$, extending the previous bound of $n\leq 5$.
Since a \swap\ gate can be implemented by $3$ \cnot\ gates, we rule out the possibility that \swap\ circuits  can be optimized by passing to \cnot\ circuits, for all $n\leq 8$.
\end{compactenum}

\Paragraph{Technical contributions.}
Towards our main contributions above, we establish a few technical results that might be of independent interest.
\begin{compactenum}
\item We establish a lower bound on the diameter of any Cayley graph as a function of its order and the sizes of its spheres at distances $1,\dots, k$, for any arbitrary $k$.
This generalizes (and strengthens) an argument made earlier for the special case of $k=1$~\cite{Patel2008}.
\item We reveal a special structure of the isometry group of $G_n$.
Intuitively, interpreting each element of $G_n$ as a \cnot\ circuit, we show that any isometry can be obtained by 
(i)~the application of a transpose-inverse map, which swaps the control and target qubits of each \cnot\ gate, followed by
(ii)~a permutation of all the qubits of the circuit.
\item For any fixed $d$, we establish polynomials in $n$ of degree $2d$, and prove that for $n \geq 2d$, they coincide with the size of the spheres of $G_n$ at distance $d$ from the identity.
\item We prove that the $3(n-p)$ conjecture for permutation matrices (see Contribution (3) above) collapses to its special case of $p=1$:~the conjecture holds for all permutations if and only if it holds for the long cycles (i.e., permutations consisting of a single cycle).
\end{compactenum}

Due to space constraints, formal proofs are delegated to the Appendix.
\section{Preliminaries}\label{sec:preliminaries}

In this section we establish general notation, and recall the well-known group structure of \cnot\ circuits~\cite{Bataille2022}.
Throughout the paper, we consider finite groups.

\subsection{A Group Structure of \cnot\ Circuits}\label{subsec:group_structure_of_cnot}

We start with a common, group-theoretic description of \cnot\ circuits.

\Paragraph{General notation.}
Given a natural number $n\in \Nats$, we let $[n]=\{1,\dots, n \}$.
A partition of $n$ is a sequence of positive natural numbers $(n_i)_i$ such that $\sum_i n_i=n$.
We primarily consider $n\times n$ matrices over the field of two elements
 $\F_2=\{0,1\}$, where addition and multiplication happen modulo $2$.
We index the rows and columns of an $n\times n$ matrix from $1$ to $n$.
We let $\IdentityMatrix_n$ be the identity $n\times n$ matrix, and let $e_i$ be the $i$-th standard basis column vector, i.e., the $i$-th column of $\IdentityMatrix_n$.

\Paragraph{Permutations.}
Let $\SymmetricGroup_n$ be the symmetric group of bijections on $[n]$.
Given a permutation $\sigma\in \SymmetricGroup_n$, we denote by $\NumPermutationCycles(\sigma)$ the number of disjoint cycles composing $\sigma$, including cycles of length 1, i.e., $\NumPermutationCycles(\sigma)$ is the length of the cycle type of $\sigma$.
We call $\sigma$ a \emph{long cycle} if $\NumPermutationCycles(\sigma)=1$.
A permutation $\sigma\in \SymmetricGroup_n$ can be represented as a permutation matrix $\PermutationMatrix_{\sigma}$ whose columns are $\PermutationMatrix_{\sigma}=[e_{\sigma(1)},...,e_{\sigma(n)}]$.
For a matrix $M$, the product $\PermutationMatrix_{\sigma}M$ is the result of permuting the rows of $M$ by $\sigma$,
while the product $M\PermutationMatrix_{\sigma}$ is the result of permuting the columns of $M$ by $\sigma^{-1}$.
The set of permutation matrices is closed under multiplication, and forms a group isomorphic to $\SymmetricGroup_n$.
The inverse of $\PermutationMatrix_{\sigma}$ is its transpose, $\PermutationMatrix^{\top}_{\sigma}=\PermutationMatrix_{\sigma}^{-1}=\PermutationMatrix_{\sigma^{-1}}$.
On individual matrix entries,
\[
(\PermutationMatrix_{\sigma}M)[i,j]=M[\sigma^{-1}(i),j] \quad \text{and } \quad (M\PermutationMatrix_{\sigma})[i,j]=M[i,\sigma(j)]
\]
which implies the following equality, that becomes useful later:
\[
(\PermutationMatrix_{\sigma}M\PermutationMatrix_{\sigma}^{-1})[i,j]=M[\sigma^{-1}(i),\sigma^{-1}(j)].
\numberthis\label{eq:permutation_conjugation}
\]

\Paragraph{Transvections.}
Consider two distinct $i,j \in [n]$, and let $\Delta_{i,j}$ be the matrix containing a single $1$ in position $(i,j)$ and which is $0$ elsewhere.
A \emph{transvection} is a matrix $\TransvectionMatrix_{i,j} = \IdentityMatrix_n + \Delta_{i,j}$.
Given a matrix $M$, the product $\TransvectionMatrix_{i,j} M$ results in adding the $j$-th row of $M$ to the $i$-th row of $M$.
In particular, if $u\in \mathbb{F}_2^n$ is a column vector representing the state of a (classical) bit-register, 
then $\TransvectionMatrix_{i,j} u$ performs the \cnot \ operation with $j$ as control and $i$ as target on the register.
Transvections enjoy the following straightforward properties (see e.g.,~\cite[Proposition~1]{Bataille2022}.)

\begin{restatable}{lemma}{lemtransvectionidentities}\label{lem:transvection_identities}
The following relations on transvections hold:
\begin{multicols}{2}
\begin{compactenum}
\item\label{item:transvection_identities1} $\TransvectionMatrix_{i,j}^2=\IdentityMatrix$.
\item\label{item:transvection_identities2} $(\TransvectionMatrix_{i,j}\TransvectionMatrix_{j,k})^2=\TransvectionMatrix_{i,k}$ for $i\neq k$.
\item\label{item:transvection_identities2.5} $(\TransvectionMatrix_{i,j}\TransvectionMatrix_{k,i})^2=\TransvectionMatrix_{k,j}$ for $j\neq k$.
\item\label{item:transvection_identities3} $(T_{i,j}T_{j,i})^2=T_{j,i}T_{i,j}$.
\item\label{item:transvection_identities4} $(\TransvectionMatrix_{i,j}\TransvectionMatrix_{k,l})^2=I$ for $i\neq l$ and $j\neq k$.
\item\label{item:transvection_identities5} $\TransvectionMatrix_{i,j}\TransvectionMatrix_{j,i}\TransvectionMatrix_{i,j}=\TransvectionMatrix_{j,i}\TransvectionMatrix_{i,j}\TransvectionMatrix_{j,i}=\PermutationMatrix_{(i,j)}$.
\end{compactenum}
\end{multicols}
\end{restatable}

\Paragraph{Transvections are generators of $\GL(n,2)$.}
We study the general linear group $\GL(n,2)$, consisting of $n\times n$ invertible matrices over $\mathbb{F}_2$.
It is known that any matrix can be brought into reduced row echelon form via elementary row operations, namely, row switching, row multiplication and row addition (e.g., by using the Gauss--Jordan algorithm). 
Since our only non-zero scalar is $1$, row multiplication is redundant, while row addition corresponds to multiplying on the left by the corresponding transvection.
Finally,  \cref{item:transvection_identities5} of \cref{lem:transvection_identities} implies that row swaps can be performed via three row additions (i.e., applying three transvections).
It thus follows that $\TransvectionGenerators_n:=\{\TransvectionMatrix_{i,j} \mid i,j\in[n], i\neq j\}$ generates $\GL(n,2)$.

\begin{wrapfigure}[10]{r}{0.5\textwidth}%
\vspace{-2.5ex}
    \centering
        \scalebox{0.8}{$\begin{tikzcd}[ampersand replacement=\&]
            \&{\begin{bmatrix}1 & 1 \\ 0 & 1 \end{bmatrix} }\arrow[r, shift right] \arrow[ld, shift right,  "T_{1,2}", swap] \&\begin{bmatrix}1 & 1 \\ 1 & 0 \end{bmatrix}\arrow[l, shift right, "T_{2,1}", swap] \arrow[rd, shift right]\\
            \begin{bmatrix}1 & 0 \\ 0 & 1 \end{bmatrix} \arrow[ru, shift right,] \arrow[rd, shift right, "T_{2,1}", swap] \& \& \& \begin{bmatrix}0 & 1 \\ 1 & 0 \end{bmatrix} \arrow[lu, shift right, "T_{1,2}", swap] \arrow[ld, shift right]\\
            \&\begin{bmatrix}1 & 0 \\ 1 & 1 \end{bmatrix} \arrow[lu, shift right] \arrow[r, shift right, "T_{1,2}", swap]\& \begin{bmatrix}0 & 1 \\ 1 & 1 \end{bmatrix} \arrow[l, shift right,] \arrow[ru, shift right, "T_{2,1}", swap]
        \end{tikzcd}$}
\caption{\label{fig:cayleygraph}
Cayley graph for $\GL(2,2)=\langle \TransvectionGenerators_2 \rangle$.
}
\end{wrapfigure}%

\Paragraph{Cayley graphs.}
Let $\Group=\langle S \rangle$ be a finite group generated by $S$.
The \emph{Cayley graph} of $\Group$ with respect to $\langle S \rangle$ is a (generally, directed) graph $G=(V, E)$, where $V=\Group$ and $E=\{(g,sg)\SuchThat g\in \Group, s\in S\}$.
We will assume throughout that the generating sets are \emph{symmetric}, i.e., that if $s \in S$, then $s^{-1} \in S$, meaning that the graph $G$ can be treated as an undirected graph.

It is useful to make a distinction between elements of $\Group$ and formal products over the generators in $S$, which are words over $S$.
We say that a word $w\in S^*$ \emph{evaluates} to $g\in \Group$ if $w$, interpreted as a product of generators, equals $g$. 
The \emph{length} of a word $w=s_1,...,s_d\in S^*$, is $d$.

Given two elements $g,h\in V$, the \emph{distance} $\Distance(g,h)$ from $g$ to $h$ is the length of a shortest path from $g$ to $h$ in $G$.
Using our notation on words, $\Distance(g,h)=d$ is the length of a shortest word $w=s_1, \dots, s_d$ such that $h=s_d\cdots s_1g$. The distance defines a metric in $G$.
With a small abuse of notation, we write $\Distance(g)$ for $\Distance(e, g)$, where $e$ is the identity of $\Group$, and refer to $\Distance(g)$ as the \emph{distance of $g$}.
The \emph{diameter} of $G$, denoted $\Diameter$, is the maximum distance between its vertices.
Given some $g\in V$ and $d\in \Nats$, the \emph{sphere} of radius $d$ centered at $g$ is the set of vertices $\Sphere(d,g)=\{ h\in V\mid \Distance(g,h)=d \}$.

\subsection{\cnot\ circuit optimization.}\label{subsec:circuit_optimization}

\Paragraph{Cayley graphs of $\GL(n,2)$.}
In this paper, we write $G_n$ for the Cayley graph of $\GL(n,2)$ with respect to the set of transvections as its generating set. See \cref{fig:cayleygraph} for a visualization for $n=2$.
Notice that, due to \cref{item:transvection_identities1} of \cref{lem:transvection_identities}, the generating set is symmetric.
As $G_n$ is vertex-transitive, its diameter can be defined as the maximum distance of a matrix $M$ from the identity $\IdentityMatrix_n$.
We write $\Diameter_n$ for the diameter of $G_n$, and $\Sphere_n(d)$ as a shorthand for the sphere $\Sphere(d, \IdentityMatrix_n)$ in $G_n$.

\Paragraph{\cnot\ circuit optimization.}
In the context of \cnot\ circuit synthesis, the following optimization question arises naturally: \emph{given some $M\in \GL(n,2)$, what is the smallest circuit (i.e., one containing the smallest number of \cnot\ gates) that implements $M$?}
It is not hard to see that the answer is the distance $\Distance(M)$, while a shortest path $\IdentityMatrix_n\Path M$ encodes such a minimal circuit for $M$.
Thus, the optimization question can be approached computationally via a BFS on $G_n$.
Note, however, that the size of $G_n$ grows super-exponentially in $n$, in particular
\[
|\GL(n,2)|=\prod_{i=0}^{n-1}(2^n-2^i)=2^{\Omega(n^2)}
\numberthis\label{eq:size_of_GL}
\]
making this approach only work for small $n$.
E.g.,~\cite{Bataille2020,Bataille2022} reports to only handle cases of $n\leq 5$.
We elevate this computational approach to handling all $n\leq 7$.

\Paragraph{The diameter of $G_n$.}
One interesting question that is also relevant to \cnot\ circuit synthesis concerns the diameter $\Diameter_n$ of $G_n$.
This captures the length of a largest optimal circuit, i.e., one that cannot be implemented with fewer \cnot\ gates.
Lower bounds on $\Diameter_n$ reveal how hard the synthesis problem can become, while upper bounds on $\Diameter_n$ confine the search space for the optimal circuit.
The computational experiments in~\cite{Bataille2022} reveal that $\Diameter_n=3(n-1)$ for all $n\leq 5$, making it tempting to assume that this pattern holds for all $n$.
This, however, is not true, as the diameter grows super-linearly in $n$~\cite{Patel2008}, in particular
\[
\Diameter_n \geq \frac{n^2-n}{\log_2(n^2-n+1)}.
\numberthis\label{eq:diameter_quadratic_lower_bound}
\]
It can be readily verified that the smallest $n$ for which the right hand side of \cref{eq:diameter_quadratic_lower_bound} becomes larger than $3(n-1)$ is $n=30$.
Since \cref{eq:diameter_quadratic_lower_bound} only states a lower bound, in \cite{Bataille2022} the following question is stated as open:
\emph{what is the value $\BreakingPoint$ of the smallest $n$ for which $\Diameter_{n}> 3(n-1)$?}
The current state of affairs places $6\leq \BreakingPoint\leq 30$.
We narrow this gap to $8\leq \BreakingPoint\leq 20$, which has half the size of the previous one.

\Paragraph{The distances of permutations.}
One notable and useful class of \cnot\ circuits is those that implement permutations $\PermutationMatrix_{\sigma}$.
\emph{Given some permutation $\sigma\in\SymmetricGroup_n$, what is the smallest circuit that implements $\PermutationMatrix_{\sigma}$?}
The following lemma gives an upper bound in terms of the number of disjoint cycles $\NumPermutationCycles(\sigma)$.

\begin{restatable}[\cite{Bataille2022}, Proposition~2]{lemma}{lempermutationdistanceupperbound}\label{lem:permutation_distance_upper_bound}
For any permutation $\sigma\in\SymmetricGroup_n$, we have that $\Distance(\PermutationMatrix_{\sigma})\leq 3(n-\NumPermutationCycles(\sigma))$.
\end{restatable}
Similarly to the computational experiments for $\Diameter_n$, in \cite{Bataille2020} it is observed that
$
\Distance(\PermutationMatrix_{\sigma})= 3(n-\NumPermutationCycles(\sigma))
$
for all $n\leq 5$ and $\sigma\in \SymmetricGroup_n$, leading to the following conjecture.

\begin{conjecture}[\cite{Bataille2020}, Conjecture~13]\label{conj:permutations}
For every $n\geq 2$ and $p\in[n]$, for every permutation $\sigma\in \SymmetricGroup_n$ with $\NumPermutationCycles(\sigma)=p$, the permutation matrix $\PermutationMatrix_{\sigma}$ lies at distance $\Distance(\PermutationMatrix_{\sigma})=3(n-p)$ in $G_n$.
\end{conjecture}

We prove (\cref{thm:permutation_distance}) that $\cref{conj:permutations}$ collapses to the case of long cycle permutations, i.e., it holds generally iff it holds for the special case of $p=1$.
Using this and computational experiments, we verify that it holds for all $n\leq 8$.

\section{BFS and the Isometries of $\GL(n,2)$}\label{sec:isometries}

In this section we present space-efficient approaches to computing distances in Cayley graphs via breadth-first traversals.
We first recall the definition of group isometries, and then equip them for space-efficient BFS traversals of Cayley graphs in a generic way.
Finally, we focus on $\GL(n,2)=\langle \TransvectionGenerators_n \rangle$, and give a precise characterization of its isometries.

\subsection{Isometries}\label{subsec:isometries}

We start by describing isometries as group automorphisms that preserve distances in the underlying Cayley graph.

\Paragraph{Group actions, orbits, and stabilizers.}
Consider a group $\Group$ and a set $X$.
Recall that a \emph{group action} ($\Group$ \emph{acting on} $X$) is a map $\cdot\colon \Group\times X\to X$ satisfying the following axioms.
\begin{compactenum}
\item \emph{(identity):~}for all $ x\in X$, we have $e\cdot x = x$, where $e $ is the identity of $\Group$.
\item \emph{(compatibility):~}for all $x\in X$ and all $g,h \in \Group$, we have $(gh)\cdot x = g\cdot (h\cdot x)$.
\end{compactenum}
Given some $x\in X$, the set $\Group\cdot x := \{g\cdot x \SuchThat g\in \Group\}$ obtained from acting with all group elements on $x$ is called the \emph{orbit of} $x$. 
The collection of all orbits $X/\Group:=\{\Group\cdot x\SuchThat x\in X\}$ partitions $X$~\cite[Theorem~2.10.5]{Lauritzen2003}. i.e., $X=\bigsqcup_{O\in X/\Group}O$.
Given some $x\in X$, the \emph{stabilizer} of $x$ is the set $\Stabilizer(x)=\{g\in \Group \SuchThat g\cdot x=x\}$ consisting of all group elements whose action on $x$ equals $x$.
This set forms a subgroup of $\Group$. 
The orbit-stabilizer theorem \cite[Theorem~2.10.5]{Lauritzen2003} together with Lagrange's theorem~\cite[Theorem~2.2.8]{Lauritzen2003} give the following relationship
\[
\lvert \Group \rvert = \lvert \Group\cdot x\rvert \lvert\Stabilizer(x)\rvert.
\numberthis\label{eq:orbit_stabilizer_size}
\]
Therefore, computing the size of the orbit of some element $x$ reduces to computing the size of the stabilizer of $x$ and the size of $\Group$.

\Paragraph{Group automorphisms.}
We will be interested in the case where the set being acted upon is itself a group.
Recall that an \emph{automorphism} of a group $\Group$ is a map $\varphi:\Group\to \Group$ that is an isomorphism from $\Group$ to itself, i.e., 
a bijection such that for all $g,h\in \Group$, we have $\varphi(gh)=\varphi(g)\varphi(h)$.
We let $\Automorphisms(\Group)$ be the set of automorphisms of $\Group$, which is itself a group under composition of maps.
The automorphism group $\Automorphisms(\Group)$ acts on $\Group$ by simple function application, i.e., for $\varphi\in \Automorphisms(\Group)$ and $g\in \Group$,
$\varphi\cdot g=\varphi(g)$, which can readily be seen to satisfy the identity and compatibility properties.

\Paragraph{Isometries and sphere partitioning.}
Consider a finite group $\Group=\langle S\rangle$ generated by a symmetric subset $S$, and let $\Distance$ be the distance map of the corresponding Cayley graph.
An automorphism $\varphi\in \Automorphisms(\Group)$ is called an \emph{isometry} (with respect to $S$) if it satisfies $\Distance(g)=\Distance(\varphi(g))$ for all $g\in \Group$.
We denote by $\Isometries(\Group)$ the set of isometries of $\Group$.
Observe that $\Isometries(\Group)$ is closed under composition, hence it is a subgroup of $\Automorphisms(\Group)$, and thus has a well-defined group action.

Consider any $d\in \Nats$, and the sphere $\Sphere(d)$ around the neutral element $e$ in the Cayley graph of $\Group$.
Since, for any isometry $\varphi\in\Isometries(\Group)$ and $g\in \Sphere(d)$, we have $\varphi(g)\in \Sphere(d)$, we may restrict our action to acting only on $\Sphere(d)$.
In particular, for any $\J\subseteq \Isometries(\Group)$, we have $\Sphere(d)=\bigsqcup_{O\in \Sphere(d)/\J}O$.
In turn, this implies that the size of the sphere $\Sphere(d)$ can be computed as the sum of the sizes of the orbits.

The following lemma captures when an automorphism $\varphi$ is an isometry.
We will use it later in \cref{subsec:isometries_in_GL} for establishing the isometries of $\GL(n,2)$.
\begin{restatable}{lemma}{lemisometriesshufflegenerators}\label{lem:isometries_shuffle_generators}
Let $\Group=\langle S\rangle$ be a finite group generated by a symmetric subset $S$. 
For any $\varphi\in\Automorphisms(\Group)$, we have $\varphi\in\Isometries(\Group)$ if and only if $\varphi(S)=S$.
\end{restatable}

\subsection{Isometry BFS}\label{subsec:isometry_BFS}
We now turn our attention to the task of traversing the Cayley graph $G$ of some finite group $\Group=\langle S \rangle$ in a breadth-first manner, given a group of isometries $\J$ of $\Group$.
Our goal is to discover the distance $\Distance(h)$ of each group element $h$, as well as the size $|\Sphere(d)|$ of each sphere of $G$.
Our technique generalizes ideas found in the literature for specific instances~\cite{Bravy-6qubit,DBLP:journals/tcad/AmyMMR13} to arbitrary groups and isometries.

Regular BFS suffers in memory the size of $\Group$, which is a bottleneck for our task of handling $\GL(n,2)$, as its size grows super-exponentially in $n$ (\cref{eq:size_of_GL}).
We address this issue by equipping BFS with isometries $\J$, which effectively allows the algorithm to store only a single representative from each orbit $\Group/\J$, thereby reducing the memory requirements.
For simplicity of presentation, for any element $h\in \Group$, we assume oracle access to
(i)~a fixed representative $\Representative(\J\cdot h)$ of the orbit of $h$, and 
(ii)~the size of the orbit $|\J\cdot h|$.
In \cref{sec:experiments}, we provide details on how we obtain this information for $\GL(n,2)$ in our experiments.

\begin{algorithm}
\DontPrintSemicolon
\KwIn{A group $\Group=\langle S\rangle$ with identity $e$ and Cayley graph $G$. A set of isometries $\J\subseteq \Isometries(\Group)$.}
\KwOut{
$\SphereSize[d]=\lvert\Sphere(d)\rvert$ and $\DistanceMap[x]=\Distance(x)$.
}
$\DistanceMap[e]\gets 0$;
$\SphereSize[0]\gets 1$\\
$\Queue.\Push(e)$\\
\While{$\Queue$ is not empty}{%
$g\gets \Queue.\Pop()$\tcp*[f]{Current vertex in the search}\\
\ForEach(\tcp*[f]{Iterate over all generators}){$s\in S$}{%
$h\gets sg$ \tcp*[f]{The $s$-successor of $g$ in $G$} \label{line:algo_isometry_bfs_get_successor}\\
$x\gets\Representative(\J \cdot h)$ \tcp*[f]{Pick the representative of the orbit of $h$} \label{line:algo_isometry_bfs_pick_representative}\\
\uIf(\tcp*[f]{The orbit of $h$ is not yet explored}){$\DistanceMap$ doesn't contain $x$}{%
$\DistanceMap[x]\gets \DistanceMap[g]+1$ \tcp*[f]{$x$ is one hop further than $g$} \label{line:algo_isometry_bfs_store_distance}\\
$\SphereSize[\DistanceMap[x]] \mathrel{+}= \lvert \J\cdot h \rvert$ \tcp*[f]{Count the size of the new orbit} \label{line:algo_isometry_bfs_update_sphere_size} \\
$\Queue.\Push(x)$ \tcp*[f]{Continue the exploration from $x$} \label{line:algo_isometry_bfs_push_x_to_queue}
}%
}%
}%
\Return{$\SphereSize[\cdot]$ and $\DistanceMap[\cdot]$}
\caption{Isometry BFS}
\label{algo:isometry_bfs}
\end{algorithm}
\Paragraph{The algorithm.}
The general description of Isometry BFS is shown in \cref{algo:isometry_bfs}.
The algorithm has the same flavor as regular BFS, using a queue $\Queue$.
However, when expanding the successors $h$ of an element $g$ (\cref{line:algo_isometry_bfs_get_successor}), it 
(i)~obtains the representative $x$ of the orbit of $h$ (\cref{line:algo_isometry_bfs_pick_representative}),
(ii)~only stores the distance of $x$ (\cref{line:algo_isometry_bfs_store_distance}),
(iii)~it increases the size of the sphere in which $x$ lies by the size of the orbit of $x$ (which is the same as the orbit of $h$, \cref{line:algo_isometry_bfs_update_sphere_size}), and
(iv)~only continues the search from $x$ (and not $h$, \cref{line:algo_isometry_bfs_push_x_to_queue}).

\Paragraph{Correctness.}
As in regular BFS, it is straightforward to see that the distances computed in $\DistanceMap[\cdot]$ are correct.
One potential threat to the correctness of the algorithm is that by only expanding the neighbours of the representatives (\cref{line:algo_isometry_bfs_push_x_to_queue}), it might not visit some vertices of $G$.
This possibility is ruled out by the following lemma.
It states that if two elements of $\Group$ are in the same orbit, then the two collections of orbits of their successors are equal.

\begin{restatable}{lemma}{lemisometrybfscorrectness}\label{lem:isometry_bfs_correctness}
Let $\Group=\langle S \rangle $ be a finite group generated by a symmetric subset $S$, and let $\J\subseteq \Isometries(\Group)$ a group of isometries of $\Group$.
For any two elements $g_1, g_2\in \Group$,  if $\J\cdot  g_1=\J\cdot g_2$ then $\{\J\cdot(sg_1)\mid s\in S\} = \{\J\cdot(sg_2)\mid s\in S\}$.
\end{restatable}

This implies that, for any $h\in \Group$, we have $\Distance(h)=\DistanceMap[\Representative(\J \cdot h)]$, allowing us to recover the distance of all elements of $\Group$ from $\DistanceMap[\cdot]$.

\begin{restatable}{theorem}{thmisometrybfsspace}\label{thm:isometry_bfs_space}
Consider an execution of \cref{algo:isometry_bfs} on a group $\Group=\langle S\rangle$ and a set of isometries $\J\subseteq \Isometries(\Group)$.
Let $G$ be the Cayley graph of $\Group$ (with respect to $S$) and $\Diameter$ the diameter of $G$.
On termination, the following hold:
\begin{compactenum}
\item For every $h\in \Group$, we have $\Distance(h)=\DistanceMap[\Representative(\J \cdot h)]$.
\item For every $d\in[\Diameter]$, we have $\lvert\Sphere(d)\rvert=\SphereSize[d]$.
\end{compactenum}
Moreover, the memory used by the algorithm is $O(\lvert\Group/\J\rvert)$.
\end{restatable}

\Paragraph{Early termination.}
We remark that \cref{algo:isometry_bfs} returns correct partial results even if it terminates early (e.g., in the case that it runs out of memory).
In particular, all distances computed in $\DistanceMap[x]$ are correct, while the size of all spheres $\SphereSize[d]$ except for the last layer are also correct.

\subsection{Isometries in $\GL(n,2)$}\label{subsec:isometries_in_GL}

\Paragraph{Symmetric group.}
Recall that $\SymmetricGroup_n$ is the symmetric group on $n$ elements.
We let $\SymmetricGroup_n$ act on $\GL(n,2)$ by conjugating with the corresponding permutation matrix, i.e.,
for $\sigma\in \SymmetricGroup_n$ and $M\in \GL(n,2)$, we have $\sigma \cdot M=\PermutationMatrix_{\sigma}M\PermutationMatrix_{\sigma}^{-1}$.
Note that conjugating by a group element is always a group automorphism.
We can show the elements of $\SymmetricGroup_n$ are isometries of $\Group$, using \cref{lem:isometries_shuffle_generators}.

\Paragraph{The transpose-inverse map.}
Let $\CyclicGroup_2=\{1,-1\}$ be the cyclic group of two elements. 
We let $\CyclicGroup_2$ act on $\GL(n,2)$ by
$
 1\cdot M=M \text{ and }(-1)\cdot M=(M^\top)^{-1}
$.
The action of $-1$ is the transpose-inverse map, which is an automorphism since
\[
-1\cdot(MN)=((MN)^{\top})^{-1}=(N^{\top} M^{\top})^{-1}=(M^{\top})^{-1}(N^{\top})^{-1}=(-1\cdot M)(-1\cdot N).
\]
Using basic properties of transvections, we can show that the elements of $\CyclicGroup_2$ are also isometries.

Observe that for $n=2$, the automorphism defined by $(1\, 2)\in \SymmetricGroup_2$ and $-1\in \CyclicGroup_2$ are equal. 
Indeed, in this case only two isometries exists, the identity map and the map defined by $\{T_{1,2}\mapsto T_{2,1}, T_{2,1} \mapsto T_{1,2}\}$. 

Automorphisms that stem from conjugation by a group element, like the group action of $\SymmetricGroup_n$, are called \emph{inner} automorphisms.
For $n\geq 3$, the transpose-inverse map is known to not be an inner automorphism of $\GL(n,2)$~\cite{TransposeInverseMapGroupProps2019}. The group actions also commute, so for $n\geq 3$, the group generated by $\SymmetricGroup_n$ and $\CyclicGroup_2$ is $\SymmetricGroup_n\times \CyclicGroup_2$.

\Paragraph{A complete characterization of $\Isometries(\GL(n,2))$.}
Finally, given our development so far, it is natural to ask~\emph{is there a succinct, syntactic characterization of all isometries in $\GL(n,2)$}?
Besides its theoretical appeal, this question also has practical implications, as working with $\J=\Isometries(\GL(n,2))$ in \cref{thm:isometry_bfs_space} leads to a more space-efficient exploration of the Cayley graph $G_n$ of $\GL(n,2)$.
As the following theorem states, the isometries in $\GL(n,2)$ are completely characterized in terms of the symmetric group and the cyclic group.

\begin{restatable}{theorem}{thmisometriescharacterizationGL}\label{thm:isometries_characterization_GL}
For any $n\geq 3$, we have that $\Isometries(\GL(n,2))=\SymmetricGroup_n\times \CyclicGroup_2$.
\end{restatable}

\section{Lower Bounds on the Diameter of $\GL(n,2)$}\label{sec:diameter_lowerbounds}

In this section we turn our attention to computing lower bounds on the diameter $\Diameter$ of the Cayley graph $G$ of a group $\Group=\langle S \rangle$.
In the context of $\GL(n,2)$, these provide a lower bound on the size of the largest \cnot\ circuit on $n$ qubits, which has gathered interest in the literature~\cite{Patel2008,Bataille2022} (see \cref{subsec:circuit_optimization}).
Although Isometry BFS reduces the memory requirements for traversing $G$, its large size can prevent the algorithm from traversing the whole graph.

\subsection{A General Inequality based on Sphere Sizes}\label{subsec:inequality}
Here we obtain an inequality that will allow us to lower-bound $\Diameter$ in terms of the sizes of the spheres $\Sphere(1),\dots,\Sphere(k)$,
where $k$ is the largest level that Isometry BFS has processed to completion.
The main idea is to bound the size of spheres at large distance (that the algorithm does not manage to compute) by the size of spheres at smaller distance, as the following lemma captures.

\begin{restatable}{lemma}{lemspheresizedecomposition}\label{lem:sphere_size_decomposition}
Let $G$ be the Cayley graph of a finite group $\Group=\langle S\rangle$ generated by a symmetric subset $S$, and let $\Diameter$ denote the diameter of the graph.
Let $d\in \NatsPlus$ with $d\leq\Diameter$, and $d_1,\dots, d_{\ell}$ be a partition of $d$.
Then $\lvert\Sphere(d)\rvert\leq \prod_{i=1}^{\ell}\lvert\Sphere(d_i)\rvert$.
\end{restatable}

Since $\lvert\Group\rvert = \sum_{d=0}^{\Diameter}\lvert\Sphere(d)\rvert$,  \cref{lem:sphere_size_decomposition} yields the following bound.

\begin{restatable}{theorem}{thmgrouporderupperbound}\label{thm:group_order_upperbound}
Let $\Group=\langle S\rangle$ be a finite group generated by a symmetric subset $S$, let $\Diameter$ denote the diameter of the corresponding Cayley graph, and let $k\in[\Diameter]$.
We have
\[
\lvert\Group\rvert \leq \sum_{d=0}^{\Diameter} \lvert\Sphere(k)\rvert^{\Quotient_k(d)}\lvert\Sphere(\Remainder_k(d))\rvert
\]
where $\Quotient_k(d)$ and $\Remainder_k(d)$ are respectively the quotient and remainder of doing integer division of $d$ by $k$.
\end{restatable}

We can now focus  on $\GL(n,2)=\langle \TransvectionGenerators_n \rangle$, and its diameter $\Diameter_n$.
\cref{thm:group_order_upperbound} generalizes an argument made in~\cite{Patel2008} for the lower bound stated in \cref{eq:diameter_quadratic_lower_bound} (using \cref{eq:size_of_GL} for $\lvert\GL(n,2)\rvert$), from $k=1$ to arbitrary $k\in[\Diameter_n]$.
This leads to tighter bounds for $\Diameter_n$ and $\BreakingPoint$, the smallest $n$ such that $\Diameter_n > 3(n-1)$ (see \cref{subsec:circuit_optimization}).
In particular, \cref{thm:group_order_upperbound} and \cref{eq:size_of_GL} yield the following corollary.

\begin{corollary}\label{cor:diameter_lowerbound_and_breakingpoint}
For any $n\in \NatsPlus$ and $k \in [\Diameter_n]$, we have $\Diameter_n \geq \ell_n(k)$, where
\[
\ell_n(k) := \min \Big\{ \ell \in \NatsPlus \; \Big| \;
\sum_{d=0}^{\ell} \lvert\Sphere_n(k)\rvert^{\Quotient_k(d)}\lvert\Sphere_n(\Remainder_k(d))\rvert \geq \prod_{i=0}^{n-1}(2^n-2^i)\Big\}\enspace.
\]
Consequently, if $\ell_n(k) > 3(n-1)$, then $n_0\leq n$.
\end{corollary}

\subsection{The Polynomial Size of Spheres in $\GL(n,2)$}\label{subsec:polynomial}

In this section, we pay attention to the rank of matrices, which we also sometimes carry explicitly in the notation.
In particular, we write a transvection of rank $n$ as $T_{i,j}^n$.
Our goal is to show that, for a fixed distance $d$, the size of the sphere $|\Sphere_n(d)|$ can be described as a polynomial in $n$, for $n$ sufficiently large (in particular, for $n\geq 2d$).
To this end we will study the orbits of our group action in more detail.
For ease of presentation, we focus primarily on orbits of the symmetric group $\SymmetricGroup_n$ acting alone, 
and consider the full isometry group $\Isometries(\GL(n,2))$ at the end.

\Paragraph{General linear subgroups of $\GL(n,2)$.}
Our first key observation is, that if $m\leq n$ then $\GL(m,2)$ is a subgroup of $\GL(n,2)$. 
This can be seen directly by considering the map $\Embedding_{m,n}\colon \GL(m,2)\to \GL(n,2)$ defined on the generators as $\Embedding_{m,n}(T_{i,j}^m)=T_{i,j}^n$,
which extends to a group homomorphism. The map can be visualized as embedding a matrix $M\in \GL(m,2)$ into the upper left corner of a larger matrix, i.e., 
\[\Embedding_{m,n}(M)=
\begin{bmatrix}
M & 0 \\
0 & \IdentityMatrix_{n-m}
\end{bmatrix}.
\]
This observation makes it clear that $\Embedding_{m,n}$ is injective.

\Paragraph{Essential indices.}
Given a matrix $M\in \GL(n,2)$ and some $i\in[n]$, we say that $i$ is an \emph{essential index} of $M$ if there exists some $j\in[n]\setminus\{i\}$ such that $M[i,j]=1$ or $M[j,i]=1$.
We let $\Essential(M)$ denote the essential indices of $M$.
Note that, for a transvection $T_{i,j}^n$, we have $\Essential(T_{i,j}^n)=\{i,j\}$.
Given a circuit $C\in \TransvectionGenerators_n^*$, we say that $C$ \emph{uses} or \emph{contains} an index $i\in[n]$, if $C$ contains a transvection in which $i$ is essential.
Next, we establish some key properties of essential indices.

The first lemma states that the essential indices of a matrix $M\in \GL(m,2)$ are preserved under the embedding $\Embedding_{m,n}$, for $m\leq n$, while a  permutation acting on $M$ permutes its essential indices.
The latter implies that all elements in the orbit $\SymmetricGroup_m \cdot M$ have the same number of essential indices.

\begin{restatable}{lemma}{lemessentialindicesinvariant}\label{lem:essential_indices_invariant}
Let $m\leq n$ and $M\in \GL(m,2)$.
The following assertions hold
\begin{compactenum}
\item\label{item:essential_invariant1} $\Essential(M)=\Essential(\Embedding_{m,n}(M))$.
\item\label{item:essential_invariant2} For each $\sigma\in \SymmetricGroup_m$, we have $\Essential(\sigma\cdot M)=\sigma(\varepsilon(M))$.
\end{compactenum}
\end{restatable}

The next lemma relates matrices of different ranks that have the same number of essential indices:~we can permute the essential indices of the higher-rank matrix to bring them to the upper left corner, making it look like the $\Embedding_{m,n}$-embedding of the lower-rank matrix.
This implies that any orbit in $\GL(n,2)/\SymmetricGroup_n$ with $m\leq n$ essential indices contains an element from the image of $\phi_{m,n}$.

\begin{restatable}{lemma}{lemcanonicalrepresentative}\label{lem:canonical_representative}
Let $N\in \GL(n,2)$ be a matrix with $\lvert\Essential(N)\rvert=m\leq n$.
Then there exists a matrix $M\in \GL(m,2)$ and a permutation $\sigma\in S_n$ such that $\Embedding_{m,n}(M)=\sigma \cdot N$.
\end{restatable}

Our third lemma is based on \cref{lem:canonical_representative} and states that essential indices of a matrix are necessary and sufficient:~every circuit evaluating to the matrix must use all its essential indices, and need not use any non-essential indices.

\begin{restatable}{lemma}{lemessentialindicesareessential}\label{lem:essential_indices_are_essential}
For any matrix $N\in \GL(n,2)$, the following assertions hold.
\begin{compactenum}
\item\label{item:essential_indeces_are_necessary} 
Any circuit $C\in \TransvectionGenerators_n^*$  that evaluates to $N$ uses all essential indices of $N$.
\item \label{item:essential_indeces_are_sufficient} 
There exists a circuit $C\in \TransvectionGenerators_n^*$ that evaluates to $N$ and uses only the essential indices of $N$.
\end{compactenum}
\end{restatable}

Since each transvection has two essential indices, \cref{lem:essential_indices_are_essential} implies that the number of essential indices of any matrix are at most twice its distance, as stated in the following lemma.
We use this observation heavily in the rest of this section.

\begin{restatable}{lemma}{lemessentialindicesatdistance}\label{lem:essential_indices_at_distance}
For any matrix $N\in \GL(n,2)$, we have $\lvert\Essential(N)\rvert\leq 2\Distance(N)$.
\end{restatable}

\Paragraph{Symmetry orbits modulo essential indices.}
Our goal is to characterize the size of the orbits of $\SymmetricGroup_n$ acting on $\GL(n,2)$.
Since elements of $\SymmetricGroup_n$ are isometries (\cref{thm:isometries_characterization_GL}), given some distance $d$, we can think of $\SymmetricGroup_n$ acting on $\Sphere_n(d)$ only, splitting it into a collection of orbits.
Let $\OrbitsOfSymmetriesOnSphere_n(d) = \Sphere_n(d)/\SymmetricGroup_n$, thus the size of the sphere at radius $d$ is $\lvert\Sphere_n(d)\rvert=\sum_{U\in \OrbitsOfSymmetriesOnSphere_n(d)}\lvert U\rvert$.
Since all elements in an orbit have the same number of essential indices (\cref{lem:essential_indices_invariant}), we proceed in a similar vein to partition the orbits of $\SymmetricGroup_n$ into parts whose elements have the same number of essential indices.
In particular, given some $m\in[n]$, we define the set of orbits $\OrbitsOfSymmetriesWithEssential_n(m) = \{  \SymmetricGroup_n \cdot M \SuchThat M\in \GL(n,2),\, \lvert\Essential(M)\rvert=m  \}$.
Finally, given $m\leq n$ and some distance $d$, let $\OrbitsOfSymmetriesOnSphereWithEssential_n(d,m) = \OrbitsOfSymmetriesOnSphere_n(d)\cap \OrbitsOfSymmetriesWithEssential_n(m)$ be the set of orbits of the sphere at distance $d$ containing matrices with $m$ essential indices.
Since the number of essential indices of a matrix is bounded by twice its distance (\cref{lem:essential_indices_at_distance}), we can write $\OrbitsOfSymmetriesOnSphere_n(d)=\bigsqcup_{m=0}^{2d} \OrbitsOfSymmetriesOnSphereWithEssential_n(d,m)$, and thus express the size of a sphere as 
\[
\lvert\Sphere_n(d)\rvert = \sum_{U\in \OrbitsOfSymmetriesOnSphere_n(d)}\lvert U \rvert = \sum_{m=0}^{2d}\left (\sum_{U\in\OrbitsOfSymmetriesOnSphereWithEssential_n(d,m)} \lvert U\rvert\right ).
\numberthis\label{eq:sphere_size_by_summing_on_essential_orbits}
\]

\Paragraph{The polynomial size of spheres.}
In order to arrive at our polynomial result, it remains to argue that the inner summation in \cref{eq:sphere_size_by_summing_on_essential_orbits} is a polynomial in $n$.
In the following, we describe in high level our strategy towards establishing this fact, while we refer to \cref{sec:app_diameter_lowerbounds} for details (see also \cref{fig:bijections} for an illustration).

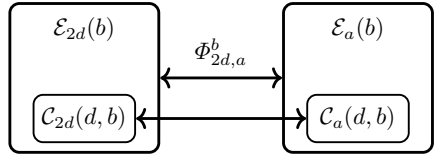
\begin{figure}
\centering
\scalebox{0.9}{
\begin{tikzpicture}[node distance=1cm and 1.8cm]
\tikzset{Eorbit/.style={rectangle, rounded corners, minimum size=2.2cm, draw=black, very thick}}
\tikzset{Corbit/.style={rectangle, rounded corners, minimum width=1.5cm, minimum height=0.7cm, draw=black, thick}}

\def\LabelDistance{-0.7cm}
\def\EYLoc{-0.6cm}

\node[Eorbit, label={[label distance=\LabelDistance]90:$\OrbitsOfSymmetriesWithEssential_{a}(b)$}] (En) at (0,0) {};
\node[Eorbit, left=of En, label={[label distance=\LabelDistance]90:$\OrbitsOfSymmetriesWithEssential_{2d}(b)$}] (E2d) {};

\node[Corbit] (Cn) at (En |-,  \EYLoc)  {$\OrbitsOfSymmetriesOnSphereWithEssential_{a}(d,b)$};
\node[Corbit] (C2d) at (E2d |-,  \EYLoc)  {$\OrbitsOfSymmetriesOnSphereWithEssential_{2d}(d,b)$};

\draw[<->, very thick] (E2d) to node[above]{$\Bijection_{2d,a}^b$} (En);

\draw[<->, very thick] (Cn) to node[above]{} (C2d);

\end{tikzpicture}
}
\medskip
\caption{\label{fig:bijections}
$\Bijection$ is a bijection between the corresponding sets.
}
\end{figure}

First, for any $a\geq b\geq c$, we establish a bijection $\Bijection_{b,a}^c$ from $\OrbitsOfSymmetriesWithEssential_{b}(c)$ to $\OrbitsOfSymmetriesWithEssential_{a}(c)$, and in particular, for any $M\in \GL(b,2)$ with $\Essential(M)=c$, we will find that $\Bijection_{b,a}^c(\SymmetricGroup_b\cdot M) = \SymmetricGroup_a \cdot \Embedding_{b,a}(M)$.
This implies that $\lvert\OrbitsOfSymmetriesWithEssential_{b}(c)\rvert=\lvert\OrbitsOfSymmetriesWithEssential_{a}(c)\rvert$, i.e., as long as we focus on orbits of matrices with $c$ essential indices, their number does not increase when we have matrices of larger order.
On the other hand, the size of each orbit $U\in \OrbitsOfSymmetriesWithEssential_{b}(c)$ increases under its image $\Bijection_{b,a}^c(U)$.
In particular, we show that
\[
\lvert\Bijection_{b,a}^c(U)\rvert=\lvert U\rvert\cdot\binom{b}{c}^{-1}\cdot \binom{a}{c}.
\numberthis\label{eq:orbit_increase}
\]
This tempts us to substitute the inner sum in \cref{eq:sphere_size_by_summing_on_essential_orbits} by
$
\sum_{U\in\OrbitsOfSymmetriesOnSphereWithEssential_m(d,m)} \lvert U\rvert\binom{n}{m}
$,
using \cref{eq:orbit_increase} for $a=n$ and $b=c=m$.
However, for this substitution to be correct, we would have to show that $\Bijection_{b,a}^b$ is also a bijection between $\OrbitsOfSymmetriesOnSphereWithEssential_{b}(d,b)$ and $\OrbitsOfSymmetriesOnSphereWithEssential_{a}(d,b)$, i.e., the distance of an orbit in $U\in \OrbitsOfSymmetriesWithEssential_{b}(b)$ does not decrease under its image $\Bijection_{b,a}^b(U)\in \OrbitsOfSymmetriesWithEssential_{a}(b)$.
We conjecture that this is indeed the case.
\begin{conjecture}\label{conj:no_shortcuts}
For any $b\leq a$ and $M\in \GL(b,2)$, we have $\Distance(M)=\Distance(\Embedding_{b,a}(M))$.
\end{conjecture}
Here we settle for a weaker statement, namely that $\Bijection_{b,a}^c$
is indeed a bijection from $\OrbitsOfSymmetriesOnSphereWithEssential_{b}(d,c)$ to $\OrbitsOfSymmetriesOnSphereWithEssential_{a}(d,c)$, provided that $b\geq 2d$.
Then, using \cref{eq:orbit_increase} in \cref{eq:sphere_size_by_summing_on_essential_orbits} for $a=n$, $b=2d$, and $c=m$, we arrive at the following theorem.

\begin{restatable}{theorem}{thmspherepolynomialsizesymmetryisometry}\label{thm:sphere_polynomial_size_symmetry_isometry}
For any fixed $d\in \Nats$, for any $n\geq 2d$,  the cardinality of $\Sphere_n(d)$ is a numerical polynomial in $n$, specifically,
\[
\lvert\Sphere_n(d)\rvert=\sum_{m=0}^{2d}\left( \sum_{U\in \OrbitsOfSymmetriesOnSphereWithEssential_{2d}(d,m)}\lvert U\rvert\cdot \binom{2d}{m}^{-1}\cdot \binom{n}{m}\right).
\]
\end{restatable}
The double sum expression in \cref{thm:sphere_polynomial_size_symmetry_isometry} is a polynomial in $n$ of degree at most $2d$. 
We also show that $\OrbitsOfSymmetriesOnSphereWithEssential_{2d}(d,2d)$ is non-empty, hence the degree is exactly $2d$.

\Paragraph{Computational implications.}
\cref{thm:sphere_polynomial_size_symmetry_isometry} directly impacts the computational use of \cref{thm:group_order_upperbound}.
In particular, when working with $\GL(n,2)$ for large $n$, Isometry BFS may fail to compute the size of a sphere $\Sphere_n(d)$, due to limited resources.
However, provided that $2d\leq n$, $|\Sphere_n(d)|$ can be calculated exactly by working in the lower-order group $\GL(2d,2)$ by
(i)~computing the sizes of the orbits $\OrbitsOfSymmetriesOnSphereWithEssential_{2d}(d,m)$ for all essential indices $m\in[2d]$, and
(ii)~using the polynomial expression in \cref{thm:sphere_polynomial_size_symmetry_isometry}.
Coming back to $\GL(n,2)$, this allows one to use larger spheres in \cref{thm:group_order_upperbound}, thereby arriving at a tighter lower bound on the diameter.

\Paragraph{Working with the full isometry group.}
It is possible to lift \cref{thm:sphere_polynomial_size_symmetry_isometry} to the full isometry group of $\GL(n,2)$, which involves the cyclic group (\cref{thm:isometries_characterization_GL}, as opposed to only the symmetric group above).
This may enable further computational approaches, as Isometry BFS (\cref{subsec:isometry_BFS}) with all isometries might scale better, thereby enabling us to obtain the sizes of spheres at larger distances.
Observe that the transpose-inverse map leaves the set of essential indices invariant for any given matrix: the actions of transposing and inverting both take non-essential indices to non-essential indices, and since the map has order two, no new non-essential indices are introduced.
In a similar fashion, we let $\OrbitsOfSymmetriesWithEssential'_n(m) = \{  (\SymmetricGroup_n\times \CyclicGroup_2) \cdot M \SuchThat M\in \GL(n,2),\, \lvert\Essential(M)\rvert=m \}$, and let $\OrbitsOfSymmetriesOnSphereWithEssential'_n(d,m) = \OrbitsOfSymmetriesOnSphere_n(d)\cap \OrbitsOfSymmetriesWithEssential'_n(m)$ be the new set of orbits of the sphere at distance $d$ containing matrices of $m$ essential indices.
We establish the following theorem.

\begin{restatable}{theorem}{thmspherepolynomialsizefullisometry}\label{thm:sphere_polynomial_size_full_isometry}
For any fixed $d\in \Nats$, for any $n\geq 2d$,  the cardinality of $\Sphere_n(d)$ is a numerical polynomial in $n$, specifically,
\[
\lvert\Sphere_n(d)\rvert=\sum_{m=0}^{2d}\left( \sum_{U\in \OrbitsOfSymmetriesOnSphereWithEssential'_{2d}(d,m)} \lvert U\rvert\cdot \binom{2d}{m}^{-1}\cdot \binom{n}{m}\right).
\]
\end{restatable}

\section{The Distance of Permutations}\label{sec:permutations}
In this section we focus on the distance of permutation matrices $P_{\sigma}$ in $G_n$, for $\sigma\in \SymmetricGroup_n$.
Recall that $P_{\sigma}$ can always be written as a product of $3(n-\NumPermutationCycles(\sigma))$ transvections (\cref{lem:permutation_distance_upper_bound}), while \cref{conj:permutations} states that this bound is tight.

We show that \cref{conj:permutations} collapses to the case of $p=1$, i.e., it holds for all permutations iff it holds for all long cycles. 
The key idea is that two disjoint cycles can be joined to one longer cycle by using one transposition (a \swap\ gate), which is the product of $3$ transvections (\cref{item:transvection_identities5} of \cref{lem:transvection_identities}).
If there is a permutation $\tau$ such that $\Distance(\PermutationMatrix_{\tau})< 3(n-\NumPermutationCycles(\tau))$, we can merge all $\NumPermutationCycles(\tau)$ cycles using $\NumPermutationCycles(\tau)-1$ transpositions, thereby constructing one long cycle of distance $<3(n-1)$.

\begin{restatable}{theorem}{thmpermutationdistance}\label{thm:permutation_distance}
\cref{conj:permutations} is true iff it holds for the special case of $p=1$.
\end{restatable}

Since $\SymmetricGroup_n$ is an isometry, \cref{thm:permutation_distance} implies that \cref{conj:permutations} collapses further to any specific long cycle permutation (e.g., $\sigma=(1\cdots n)$).

\section{Experimental Results}\label{sec:experiments}

Here we report experimental results based on the theoretical development above.

\Paragraph{Implementation.}
We have implemented Isometry BFS (\cref{algo:isometry_bfs}) for $\GL(n,2)=\langle\TransvectionGenerators_n\rangle$ using the symmetric group $\SymmetricGroup_n$ as the isometry $\J$.
Interpreting a matrix $h\in \GL(n,2)$ as a graph, its orbit $\J \cdot h$ corresponds to isomorphic graphs.
We utilize the \texttt{nauty}-software~\cite{McKay2014} on graph isomorphism to compute the representative $\Representative(\J \cdot h)$ and the size of the orbit $|\J \cdot h|$ during the exploration.
We do not use the cyclic group $\CyclicGroup_2$ in $\J$ as it requires inverting a matrix, which is a time-consuming operation in general, while its best-case effect would be to halve the memory requirements.
To achieve parallel speedup, we store all representatives of a level in a lock-free concurrent hash table, following the design in \cite{DBLP:conf/fmcad/LaarmanPW10} and using an implementation from \cite{DBLP:conf/nfm/Berg21}. The elements of the level are enumerated and processed in parallel (i.e., each worker takes some batches from the current BFS level), relying on OpenMP.

\Paragraph{Setup.}
We run our experiments on a large 40-core machine (Intel Xeon Gold 6230) with 1.5TB of internal memory and a 2.1GHz clock.
Although we do not report on precise timing measurements, we note that the largest experiment mentioned here was completed in 5.5 hours (wall clock time).

\Paragraph{Cayley graphs for $n=6, 7$.}
Using the above setup, we have performed a full exploration of $G_1,\dots, G_7$.
\cref{tab:symmetryReduction} gives an indication of the memory savings obtained by symmetry reduction.
We find that $\Diameter_6=15$ and $\Diameter_7=18$, confirming that the $\Diameter_n=3(n-1)$ for $n=6,7$.
This implies that $\BreakingPoint\geq 8$, tightening the previous bound of $\BreakingPoint\geq 6$.

We refer to \cref{app:tables} for more details on the computation.
Here we just mention that the largest sphere at $n=7$, $d=14$, i.e. $\Sphere_7(14)$, contains
68,493,002,803,224 elements (\cref{tab:levelsizes}), 
represented by $\lvert\Sphere_7(14)/{\SymmetricGroup_7}\rvert = 13,616,116,190$ orbits (\cref{tab:orbits}).
This largest BFS level was stored in a concurrent hash-table of size $2^{35}$ nodes.
\begin{table}[!b]
\caption{Sizes of $\GL(n, 2)$ and their
symmetry-reduced versions.\label{tab:symmetryReduction}}
\adjustbox{max width=\textwidth}{
\begin{tabular}{|c||r|r|r|r|r|r|r|}
\hline
$n$ & 
\multicolumn{1}{c|}{\;1\;} & 
\multicolumn{1}{c|}{\;2\;} & 
\multicolumn{1}{c|}{3} & 
\multicolumn{1}{c|}{4} & 
\multicolumn{1}{c|}{5} & 
\multicolumn{1}{c|}{6} & 
\multicolumn{1}{c|}{7}
\\
\hline
$\lvert\GL(n, 2)\rvert$ &
1 & 6 & 168 & 20,160 & 9,999,360 & 20,158,709,760 & 163,849,992,929,280\\
$\lvert\GL(n, 2)/{\SymmetricGroup_n}\rvert$ & 
1 & 4 & 33 & 908 & 85,411 & 28,227,922 & 32,597,166,327\\
\hline
\end{tabular}
}
\end{table}

\Paragraph{The distance of permutations for $n=6,7,8$.}
We have also verified that \cref{conj:permutations} holds for all $n\leq 8$, i.e., for every permutation $\sigma\in \SymmetricGroup_n$, we have $\Distance(\PermutationMatrix_{\sigma})=3(n-\NumPermutationCycles(\sigma))$. 
Although this was straightforward for $n=6,7$, since we could compute the whole Cayley graph, the case of $n=8$ was more challenging.
Here, Isometry BFS only succeeded in $12$ levels due to memory limitations.
To circumvent this, we performed a bi-directional search~\cite{DBLP:journals/tcad/AmyMMR13,DBLP:conf/mfcs/AlonGJL24}.
We computed 12 forward levels $F_i$ for $i\in[12]$ from $\IdentityMatrix_8$, and 9 backward levels $B_i$ for $i\in[9]$ from the permutation matrix $\PermutationMatrix_{\sigma}$, where $\sigma=(1\cdots 8)$.
$F_{12}$ and $B_9$ contain 33,719,514,377 and 65,936,050,032 orbits,
resp., requiring a concurrent hash-table of $2^{36}$ nodes.

We confirmed that the two searches did not discover a common element
before $F_{12}\cap B_{9}$, which means that $\Distance(\PermutationMatrix_{\sigma})\geq 21 = 3(n-\NumPermutationCycles(\sigma))$.
This, together with \cref{lem:permutation_distance_upper_bound} and \cref{thm:permutation_distance}, concludes that \cref{conj:permutations} holds for all $n\leq 8$, increasing the previous bound of $n\leq 5$.

\begin{table}[t!]
\caption{Coefficients $a_{d,m}$ for $d\in[10]$, such that $f_d(n) = \sum_{m=0}^{2d}a_{d,m} \binom{n}{m}$.\medskip
\label{tab:polynomials}}
\small
\adjustbox{max width=\textwidth}{
\begin{tabular}{|r|cccccccccc|}
\hline
  & \multicolumn{10}{|c|}{$d$}\\
$m$ & 1 & 2 & 3 & 4 & 5 & 6 & 7 & 8 & 9 & 10\\
\hline
0 & 0 & 0 & 0 & 0 & 0 & 0 & 0 & 0 & 0 & 0 \\
1 & 0 & 0 & 0 & 0 & 0 & 0 & 0 & 0 & 0 & 0 \\
2 & 2 & 2 & 1 & 0 & 0 & 0 & 0 & 0 & 0 & 0 \\
3 & - & 18 & 48 & 60 & 24 & 2 & 0 & 0 & 0 & 0 \\
4 & - & 12 & 344 & 1818 & 5220 & 7522 & 4058 & 541 & 6 & 0 \\
5 & - & - & 360 & 9990 & 91520 & 502840 & 1749420 & 3568470 & 3225280 & 736540 \\
6 & - & - & 120 & 13200 & 398952 & 5617980 & 51420950 & 333774990 & 1541881070 & 4749327810 \\
7 & - & - & - & 7560 & 601020 & 20575002 & 405060894 & 5567304106 & 57957418260 & 470186283084 \\
8 & - & - & - & 1680 & 456960 & 33322352 & 1307932768 & 33637869692 & 641405868096 & 9693333049694 \\
9 & - & - & - & - & 181440 & 30285360 & 2201160528 & 98951246910 & 3171772301544 & 79064742058728 \\
10 & - & - & - & - & 30240 & 16380000 & 2257118640 & 169797210840 & 8680734360440 & 335405245663920 \\
11 & - & - & - & - & - & 4989600 & 1491890400 & 189509942160 & 15030104274900 & 866095057466166 \\
12 & - & - & - & - & - & 665280 & 629354880 & 144347464800 & 17833379314080 & 1504147346394528 \\
13 & - & - & - & - & - & - & 155675520 & 75463708320 & 15090657341760 & 1867284211941720 \\
14 & - & - & - & - & - & - & 17297280 & 26153487360 & 9209014214400 & 1712722052801760 \\
15 & - & - & - & - & - & - & - & 5448643200 & 3994763572800 & 1175390846229600 \\
16 & - & - & - & - & - & - & - & 518918400 & 1176215040000 & 600544748643840 \\
17 & - & - & - & - & - & - & - & - & 211718707200 & 222992728358400 \\
18 & - & - & - & - & - & - & - & - & 17643225600 & 57111121267200 \\
19 & - & - & - & - & - & - & - & - & - & 9050974732800 \\
20 & - & - & - & - & - & - & - & - & - & 670442572800 \\
\hline
\end{tabular}
}
\end{table}

\Paragraph{New lower bounds on $\Diameter_n$ and $\BreakingPoint$.}
By instrumenting our implementation, we computed
the coefficients of the numeric polynomials $f_1(n),\ldots,f_{10}(n)$,
such that $f_d(n)=\lvert\Sphere_n(d)\rvert$ (for $n\geq 2d$),
following \cref{thm:sphere_polynomial_size_symmetry_isometry}.
To compute the coefficients of the polynomial $f_{10}(n)$ of degree 20, we need to compute BFS levels up to $\Sphere_{20}(10)$,
i.e., all \cnot\ circuits on 20 qubits of size 10. 
The last level contains $1.7\times 10^{19}$ elements, represented by ``only'' $7.4\times 10^8$ orbits.
We kept counts of the elements with $0,\ldots,20$ essential indices.

\cref{tab:polynomials} reports the coefficients $a_{d,m}$ of these ten polynomials, where $f_d(n) = \sum_{m=0}^{2d}a_{d,m} \binom{n}{m}$. 
For instance, the first three polynomials read as follows:
\begin{eqnarray*}
\nonumber
\scriptstyle f_1(n)\, =
    & \scriptstyle \underline{2} \binom{n}{2} 
    & \scriptstyle =\, 1(n^2 - n), \\
\nonumber
\scriptstyle f_2(n)\, = 
    & \scriptstyle \underline{2} \binom{n}{2} + \underline{18}\binom{n}{3} + \underline{12} \binom{n}{4} 
    & \scriptstyle =\, \frac{1}{2}(n^4 - 5n^2 + 4n),\\
\nonumber
\scriptstyle f_3(n)\, =
  & \scriptstyle \underline{1} \binom{n}{2} + \underline{48}\binom{n}{3} + \underline{344} \binom{n}{4} + \underline{360} \binom{n}{5} + \underline{120} \binom{n}{6}
  & \scriptstyle =\, \frac{1}{6}(n^6 + 3 n^5 -9 n^4 -63 n^3 + 179 n^2 -111 n).
\nonumber
\end{eqnarray*}
We only proved $f_d(n)=|\Sphere_n(d)|$ for $n\geq 2d$, but one can readily check that this equation holds
for all $1\leq d\leq 10$ and $n\leq 8$,
as predicted by \cref{conj:no_shortcuts}.
We can now use \cref{cor:diameter_lowerbound_and_breakingpoint} for $k=10$ to compute the lower bound $\ell_n(10)$ of $\Diameter_n$.
\begin{table}[!b]
\caption{Computed lower bounds on the diameter, $\ell_n(10)\leq \Diameter_n$\label{tab:lowerbounds}}
\centering
\begin{tabular}{|c||c|c|c|c|c|c|c|c|c|c|c|c|c|}
\hline
$n$      & 20 & 21 & 22 & 23 & 24 & 25 & 26 & 27 & 28 & 29 & 30 & \ldots & 40 \\ 
$\ell_n(10)$ & 58 & 63 & 68 & 73 & 78 & 83 & 89 & 95 & 101 & 107 & 113 & \ldots & 183\\
\hline
\end{tabular}
\end{table}
In \cref{tab:lowerbounds}, we compute $\ell_{20}(10),\ldots,\ell_{30}(10)$ and $\ell_{40}(10)$. 
We find that $\ell_{20}(10)=58 > 57 = 3(20-1)$, so $\BreakingPoint\leq 20$, i.e., there is an optimal circuit on $n=20$ qubits with length $>3(n-1)$.
We also computed $\ell_{40}(10)$ as a witness that for $n=40$, an optimal circuit of length beyond $4n$ exists.

\section{Conclusion}\label{sec:conclusion}
In this paper, we have developed group-theoretic techniques to address questions in optimal synthesis of \cnot\ circuits, concerning 
(i)~the exact sizes of optimal circuits that perform a given function,
(ii)~the size of the largest optimal circuit for a given number of qubits $n$, and
(iii)~the sizes of permutation circuits.

An interesting direction for future research includes extending our approach to larger gate sets (e.g., $\cnot + \mathtt{T}$~\cite{DBLP:conf/rc/MeuliSM18}, or the Clifford fragment~\cite{Bravy-6qubit}).
Other directions are to incorporate hardware layout restrictions,
where not all \cnot-operations are allowed, or to relax the problem by allowing any permutation of the output qubits, or to minimize for circuit depth rather than circuit size.
Optimal synthesis for these cases by means of SAT solving is proposed in~\cite{DBLP:conf/ecai/ShaikP24}. After this paper was submitted,
\cite{Webster2025} shows how a BFS based on canonical graphs can also be used for various of these extensions. 
But computing the size (or depth) of the longest optimal circuit for a given number of qubits is still open for all these cases.

\subsection*{Acknowledgement}
The numerical results were obtained at the Grendel cluster of the Centre for Scientific Computing, Aarhus \url{https://phys.au.dk/forskning/faciliteter/cscaa/}. The research is partially funded by the Innovation Fund Denmark through the project ``Automated Planning for Quantum Circuit Optimization'', and by the European Innovation Council through Accelerator grant no. 190124924.
\bibliographystyle{plain}
\bibliography{bibliography, bibliography-tmp}

\newpage
\appendix
\section{Proofs of \cref{sec:isometries}}\label{sec:app_isometries}

\subsection{Proofs of \cref{subsec:isometries}}

\lemisometriesshufflegenerators*
\begin{proof}
First, assume that $\varphi\in \Isometries(\Group)$. 
An element of $\Group$ has distance 1 if and only if it is a generator. 
Hence, for any generator $s\in S$ both $\varphi(s)$ and $\varphi^{-1}(s)$ must be generators. 
Note that $\varphi^{-1}(s)$ exists since $\varphi$ is bijective. 
We thus have that $\varphi(S)=S$.

Conversely, suppose that $\varphi(S)=S$ and let $g\in G$. 
Note that
\[
 \text{for all $s_1,\dots,s_d\in S$ we have } g=\prod_{i\in[d]} s_i \implies \varphi(g)= \prod_{i\in[d]} \varphi(s_i),
\]
which proves that $\Distance(g) \geq \Distance(\varphi(g))$,
since any word evaluating to $g$ gives us a word evaluating to $\varphi(g)$ of the same length, using the fact $\varphi(s_i)\in S$. 
But we also have that
\[
\text{for all $t_1,\dots,t_d\in S$ we have } \varphi(g)=\prod_{i\in[d]} t_i \implies g= \prod_{i\in[d]} \varphi^{-1}(t_i),
\]
proving that $\Distance(g) \leq \Distance(\varphi(g))$ since $\varphi^{-1}(t_i)\in S$ as well. 
Thus, $\Distance(g)=\Distance(\varphi(g))$, so $\varphi$ is an isometry.
\qed
\end{proof}

\subsection{Proofs of \cref{subsec:isometry_BFS}}

\lemisometrybfscorrectness*
\begin{proof}
Because of symmetry, it suffices to show just one inclusion. 
Take a generator $s\in S$ and consider the orbit $\J\cdot(sg_1)$. 
It suffices to show there exists a $t\in S$ such that $sg_1\in \J(tg_2)$. 
Since $g_1$ and $g_2$ are in the same orbit, there exists an isometry $\varphi\in \J$ such that $\varphi(g_2)=g_1$.
Thus choosing $t=\varphi^{-1}(s)$ (which is an element of $S$ thanks to \cref{lem:isometries_shuffle_generators}) will work because $\varphi(tg_2)=\varphi(t)\varphi(g_2)=sg_1$.
The desired result follows.
\qed
\end{proof}

\subsection{Proofs of \cref{subsec:isometries_in_GL}}

\begin{restatable}{lemma}{lemsymmetricgroupisometries}\label{lem:symmetric_group_isometries}
For all $n\geq 2$, we have $\SymmetricGroup_n\subseteq \Isometries(\GL(n,2))$.
\end{restatable}
\begin{proof}
Recall our definition of permutation matrices in \cref{subsec:group_structure_of_cnot}, where we noted that the $(i,j)$-th entry of $M$ equals the $(\sigma(i),\sigma(j))$-th entry of $\PermutationMatrix_\sigma M\PermutationMatrix_\sigma^{-1}$. 
In particular, for a transvection $T_{i,j}$, we have $\sigma\cdot T_{i,j}=T_{\sigma(i),\sigma(j)}$.
Intuitively, this group action simply relabels the bits of our register according to $\sigma$.
Since transvections form a generating set of $\GL(n,2)$, we can use \cref{lem:isometries_shuffle_generators} to get statement of the lemma.
\qed
\end{proof}

\begin{restatable}{lemma}{lemcyclicgroupisometries}\label{lem:cyclic_group_isometries}
For all $n\geq 2$, we have $\CyclicGroup_2\subseteq \Isometries(\GL(n,2))$.
\end{restatable}
\begin{proof}
First, observe that, for a transvection $T_{i,j}$,  we have $T_{i,j}^{\top}=T_{j,i}$.
By \cref{item:transvection_identities1} of \cref{lem:transvection_identities}, we have $T_{i,j}^{-1}=T_{i,j}$, and thus $-1\cdot T_{i,j}=T_{j,i}$.
Since transvections form a generating set of $\GL(n,2)$, \cref{lem:isometries_shuffle_generators} yields the statement of the lemma.
\qed
\end{proof}
\begin{restatable}{lemma}   {lemcommutativitygroupaction}\label{lem:commutativegroupaction}
The actions of $\SymmetricGroup_n$ and $\CyclicGroup_2$ on $\GL(n,2)$ commute, i.e. for every $M\in \GL(n,2)$, $\sigma\in \SymmetricGroup_n$ and $\xi\in \CyclicGroup_2$ it holds that $\xi\cdot (\sigma\cdot M)=\sigma \cdot (\xi\cdot M) $.
\end{restatable}
\begin{proof}
    The result follows immediately from the fact that the actions commute on transvections. The interesting case is when $\xi=-1$ where we see
    \[
    \begin{tikzcd}
        &T_{i,j} \arrow[r,"\sigma", maps to] \arrow[d,"-1", maps to] &T_{\sigma(i),\sigma(j)} \arrow[d,"-1", maps to] \\
        &T_{j,i} \arrow[r,"\sigma", maps to] \arrow[r,"\sigma", maps to] &T_{\sigma(j),\sigma(i)}
    \end{tikzcd}
    \]
    using the observations from the proofs of \cref{lem:symmetric_group_isometries} and \cref{lem:cyclic_group_isometries}.
\qed
\end{proof}

The next lemma is a stepping stone towards \cref{thm:isometries_characterization_GL} afterwards.

\begin{restatable}{lemma}{lemcontroltargetswapstability}\label{lem:control_target_swap_stability}
Fix some $n\geq 3$ and $\psi\in \Isometries(\GL(n,2))$.
Consider any two distinct indices $i,j\in [n]$, and let $\psi(T_{i,j})=T_{a,b}$, for some $a,b\in[n]$.
Then $\psi(T_{j,i})=T_{b,a}$.
\end{restatable}
\begin{proof}
To simplify the notation, let $\TVShort{i,j}$ denote the transvection $T_{i,j}$. 
We also write $\psi\TVShort{i,j}$ to mean $\psi(\TVShort{i,j})$ i.e. the image of $\TVShort{i,j}$ under $\psi$, to avoid the cluttering extra parentheses.
Since $\psi\TVShort{i,j}=\TVShort{a,b}$, using \cref{item:transvection_identities3} of \cref{lem:transvection_identities}, we have
\[
\psi(\TVShort{j,i}\TVShort{i,j})=\psi((\TVShort{i,j}\TVShort{j,i})^2)=(\psi\TVShort{i,j}\psi\TVShort{j,i})^2=(\TVShort{a,b}\psi\TVShort{j,i})^2.
\numberthis
\label{eq:two_hop_distances}
\]
Let $M=\TVShort{j,i}\TVShort{i,j}$, and observe that the $(j,j)$-th entry of $M$ is $0$, which means $M$ is neither the identity matrix nor a generator, and thus $\Distance(M)=2$.
Since $\psi$ is an isometry, we obtain that $(\TVShort{a,b}\psi\TVShort{j,i})^2$ must also have distance 2, due to \cref{eq:two_hop_distances}.
We now have $\psi\TVShort{j,i}=\TVShort{b,a}$ since for any other indices, \cref{item:transvection_identities2}, \cref{item:transvection_identities2.5} and \cref{item:transvection_identities4} of \cref{lem:transvection_identities} would imply that the product $(\TVShort{a,b}\psi\TVShort{j,i})^2$ has distance $0$ or $1$.
The desired result follows.
\qed
\end{proof}

\thmisometriescharacterizationGL*
\begin{proof}
In this proof, we use the simplified notation from the proof of \cref{lem:control_target_swap_stability}.
We will argue that  for any $\psi\in\Isometries(\GL(n,2))$, there exists $(\sigma,\xi)\in \SymmetricGroup_n\times \CyclicGroup_2$ such that $\psi(M)=(\sigma,\xi)\cdot M$ for all $M\in \GL(n,2)$.

Consider three distinct numbers $i,j,k\in [n]$. 
By \cref{lem:isometries_shuffle_generators}, $\psi$ must map generators to other generators. We analyse what happens to transvections with $i$ as target.
Let $\psi\TVShort{i,j}=\TVShort{a,b}$ and $\psi\TVShort{j,k}=\TVShort{c,d}$, for $a,b,c,d\in [n]$.
Since $\psi$ is an automorphism, we have 
\[
\psi\TVShort{i,k}=\psi((\TVShort{i,j}\TVShort{j,k})^2)=(\TVShort{a,b}\TVShort{c,d})^2.
\numberthis\label{eq:control_target_swap_stability_eq1}
\]
Due to \cref{item:transvection_identities2} and \cref{item:transvection_identities4} of \cref{lem:transvection_identities}, \cref{eq:control_target_swap_stability_eq1} implies that either
(1)~$a\neq d$ and $b=c$, or
(2)~$a=d$ and $b\neq c$.
We examine each case.

\begin{compactenum}
\item Assume that $a\neq d$ and $b=c$.
Continuing on \cref{eq:control_target_swap_stability_eq1}, we have
\[
\psi\TVShort{i,k}=(\TVShort{a,b}\TVShort{b,d})^2=\TVShort{a,d}.
\]
We see for both $\psi\TVShort{i,k}$ and $\psi\TVShort{i,j}$, with $i$ as the target index in the input we have $a$ as the target index in the output. 

We will argue that this is always the case, i.e., $\psi\TVShort{i,\ell}=\psi\TVShort{a,q}$ for all $\ell\in[n]$ and some $q\in [n]$.
Towards this, assume that $\ell$ is different from $i,j,k$, and let $\psi\TVShort{j,\ell}=\TVShort{p,q}$.
Then, we may once again calculate
\[
\psi\TVShort{i,\ell} = \psi((\TVShort{i,j}\TVShort{j,\ell})^2) = (\TVShort{a,b}\TVShort{p,q})^2
\]
which, in turn, implies that either
(i)~$a\neq q$ and $b=p$, or
(ii)~$a=q$ and $b\neq p$.
Observe that case (i) proves our claim.
Assume for contradiction that case (ii) holds.
We then have
\[
\IdentityMatrix_n=\psi(\IdentityMatrix_n)=\psi((\TVShort{i,k}\TVShort{j,\ell})^2)=(\TVShort{a,b}\TVShort{p,q})^2=(\TVShort{a,b}\TVShort{p,a})^2=\TVShort{p,b}
\]
which is clearly false.
Thus case (i) holds, concluding our claim that whenever $i$ is the target index of a transvection $T$, $a$ is the target index of $\psi(T)$. 
By \cref{lem:control_target_swap_stability}, we then also have that whenever whenever $i$ is the control index of a transvection $T$, $a$ is the control index of $\psi(T)$. 
For arbitrary $j,k\in [n]$ if we suppose $\psi\TVShort{i,j}=\TVShort{a,b}$ and $\psi\TVShort{i,k}=\TVShort{a,c}$ we then get
\[
\psi\TVShort{j,k}=\psi((\TVShort{j,i}\TVShort{i,k})^2)=(\TVShort{b,a}\TVShort{a,c})^2=\TVShort{b,c}.
\]
Thus, for any $j\in [n]$ there exists a unique element $\sigma(j) \in [n]$, such that when $j$ is the target of a transvection $T$, $\sigma(j)$ is the target of $\psi (T)$. 
\\Using \cref{lem:control_target_swap_stability} once again, this defines a map $\sigma\colon [n] \to [n]$ with the property that $\psi(T_{i,j})=T_{\sigma(i),\sigma(j)}$ for every $i,j\in [n]$. 
This $\sigma$ must be injective since $\sigma(i)=\sigma(j)$ for $i\neq j$ would imply that $\psi\TVShort{i,j}=[\sigma(i),\sigma(j)]$ would not be a well-defined transvection.
We conclude that $\sigma$ is a permutation of $[n]$ and indeed $\psi\TVShort{i,j}=(\sigma,1)\cdot \TVShort{i,j}$.

\item Assume that $a=d$ and $b\neq c$.
The proof for this case proceeds similarly as case (1), with the conclusion that $\psi(T_{i,j})=(\sigma,-1)\cdot T_{i,j}$ for some permutation $\sigma\in \SymmetricGroup_n$.
\end{compactenum}
\qed
\end{proof}
\section{Proofs of \cref{sec:diameter_lowerbounds}}\label{sec:app_diameter_lowerbounds}

\subsection{Proofs of \cref{subsec:inequality}}

\lemspheresizedecomposition*
\begin{proof}
Define a map
\[
    f\colon \bigtimes_{i=1}^{\ell} \Sphere(d_i) \to \Group.
\]
by $f(g_1, \dots, g_\ell) = g_1 \cdots g_\ell$. We claim that $R(d)$ is a subset of the image of $f$, from which it follows that
\[
    \lvert R(d) \rvert \leq \lvert \mathrm{im}(f)\rvert \leq \left\lvert \bigtimes_{i=1}^{\ell} \Sphere(d_i) \right\rvert = \prod_{i=1}^\ell \lvert R(d_i) \rvert.
\]
To see this, let $g \in R(d)$ and write $g = s_1 s_2 \cdots s_d$ for $s_i \in S$. Now, for $j = 1, \dots, \ell$, define $g_j = s_{a_j} s_{a_j+1} \cdots s_{b_j}$, where  $a_j = 1 + \sum_{k=1}^{j-1} d_j$ and $b_j = \sum_{k=1}^j d_j$. Then $g = g_1 \cdots g_j$. Moreover, $g_j \in \Sphere(d_j)$ since the definition of $g_j$ uses at exactly $d_j$ generators, and no shorter products of generators would work, since otherwise we would be able to write $g$ as a product of strictly less than $d$ generators.
\qed
\end{proof}

\thmgrouporderupperbound*
\begin{proof}
First, write the order of the group as 
$
|\Group|=\sum_{d=0}^{\Diameter}|\Sphere(d)|
$
and then apply \cref{lem:sphere_size_decomposition} on each summand $\lvert\Sphere(d)\rvert$ using the partition of $d$ given by $d=\Quotient_k(d)\cdot k + \Remainder_k(d)$.
\qed
\end{proof}

\subsection{Proofs of \cref{subsec:polynomial}}

\lemessentialindicesinvariant*
\begin{proof}
\begin{compactenum}
\item This is immediate from the definition of the embedding $\Embedding_{m,n}$.
\item The result follows from \cref{eq:permutation_conjugation}, which says $M[i,j]=(\sigma \cdot M)[\sigma(i),\sigma(j)]$. Thus the left hand side is $1$ for $j\in [n]\setminus\{i\}$ if and only if the right hand side is. 
\end{compactenum}
\qed
\end{proof}

\lemcanonicalrepresentative*
\begin{proof}
Since $\Essential(N)=m\leq n$, we can establish an injective map $\tau\colon[m]\to [n]$ such that the image of $\tau$ is $\Essential(N)$.
Extend $\tau$ to a permutation $\sigma\colon [n]\to [n]$ by having the indices $i\in\{m+1,\dots,n\}$ map to the non-essential indices of $N$, of which we have $|[n]-\Essential(N)|=n-m$. 
By construction $\sigma\cdot N$ is now of the form
\[
\begin{bmatrix}
M & 0 \\
0 & \IdentityMatrix_{n-m}
\end{bmatrix}
\]
for some matrix $M$.
Note that $M$ must be invertible, thus $M\in \GL(m,2)$, yielding that $\Embedding_{m,n}=\sigma\cdot N$, as desired.
\qed
\end{proof}

\lemessentialindicesareessential*
\begin{proof}
\begin{compactenum}
\item 
Let $N$ be the matrix that $C$ evaluates to, and consider any index $i\in [n]$ that $C$ does not use. We will argue that $i\not \in \Essential(N)$.
Indeed, the $i$-th row of $N$ is $e_i^\top$, since evaluating $C$ will never involve adding another row to the $i$-th one. 
Similarly, the $i$-th column of $N$ is equal to $e_i$ since evaluating $C$ will never involve adding the $i$-th row to another row.
Hence $i\not \in \Essential(N)$.
\item
Let $m=|\Essential(N)|$.
By \cref{lem:canonical_representative}, there exists a permutation $\sigma\in \SymmetricGroup_n$ and a matrix $M\in \GL(m,2)$ such that $\Embedding_{m,n}(M)=\sigma \cdot N$.
Let $C\in \TransvectionGenerators_m^*$ be a circuit generating $M$.
We obtain a circuit $C'\in \TransvectionGenerators_n^*$ by simply applying $\Embedding_{m,n}$ on each transvection of $C$.
By definition, $\Embedding_{m,n}$ maps transvections to transvections without changing the indices used, thus $C'$ only uses indices from $[m]$.
Finally, we obtain a circuit $C''$ that evaluates to $N$ by acting with $\sigma^{-1}$ on each transvection of $C'$.
Observe that $C''$ uses only the essential indices of $N$, as desired.
\end{compactenum}
\qed
\end{proof}

\lemessentialindicesatdistance*
\begin{proof}
Consider any circuit $C$ of length $d$ that evaluates to $N$.
Due to \cref{item:essential_indeces_are_necessary} of \cref{lem:essential_indices_are_essential}, $C$ uses all essential indices of $N$.
Moreover, since a transvection uses $2$ indices, we have that $C$ uses at most $2d$ indices.
Thus $\lvert\Essential(N)\rvert\leq 2\Distance(N)$.
\qed
\end{proof}

\Paragraph{The polynomial size of spheres.}
Here we make the arguments behind \cref{thm:sphere_polynomial_size_symmetry_isometry} formal.
We begin with a lemma that the embedding $\Embedding_{m, n}$ preserves symmetry orbits.

\begin{restatable}{lemma}{lemorbitsofessentialindicesareunique}\label{lem:orbits_of_essential_indices_are_unique}
Let $0\leq m\leq n$ and $M_1,M_2\in \GL(m,2)$. 
Then $M_1\in \SymmetricGroup_m\cdot M_2$  if and only if  $\Embedding_{m,n}(M_1)\in \SymmetricGroup_n\cdot \Embedding_{m,n}(M_2)$.
\end{restatable}
\begin{proof}
First, note that $\SymmetricGroup_m$ is a subgroup of $\SymmetricGroup_n$ in the natural way that any permutation $\sigma\in\SymmetricGroup_m$ can be extended to a permutation $\sigma'\in \SymmetricGroup_n$ by mapping each index $i$ with $m+1\leq i\leq n$ to itself.

Now, assume that $M_1\in \SymmetricGroup_m \cdot M_2$, thus there exists some $\sigma \in \SymmetricGroup_m$ such that $M_1=\sigma \cdot M_2$.
Then, defining $\sigma'$ as the natural extension of $\sigma$ to $n$ indices, we have $\Embedding_{m,n}(M_1)=\sigma' \cdot \Embedding_{m,n}(M_2)$.

For the opposite direction, assume that $\Embedding_{m,n}(M_1)\in \SymmetricGroup_n\cdot \Embedding_{m,n}(M_2)$, thus there exists some $\sigma\in \SymmetricGroup_n$ such that $\Embedding_{m,n}(M_1)=\sigma\cdot\Embedding_{m,n}(M_2)$.
By \cref{item:essential_invariant1} of \cref{lem:essential_indices_invariant}, we have $\Essential(M_i)=\Essential(\Embedding_{m,n}(M_i))$, for each $i\in[2]$.
By \cref{item:essential_invariant2} of \cref{lem:essential_indices_invariant}, we have
\[
\Essential(M_1)=\Essential(\Embedding_{m,n}(M_1))=\Essential(\sigma\cdot \Embedding_{m,n}(M_2))=\sigma(\Essential(M_2)),
\]
thus $\sigma$ restricted to $\Essential(M_2)$ has image $\Essential(M_1)$.
This restricted $\sigma$ can then be extended to a mapping $\sigma'\colon[m]\to [m]$ that maps the non-essential indices of $M_2$ to non-essential indices of $M_1$.
Observe that $M_1 = \sigma' \cdot M_2$, and thus  $M_1\in \SymmetricGroup_m\cdot M_2$, as desired.
\qed
\end{proof}

\cref{lem:canonical_representative} and \cref{lem:orbits_of_essential_indices_are_unique} imply that for $m\leq n$ and a given number of essential indices $k\in[m]$, there is a well-defined bijection $\Bijection_{m,n}^k\colon \OrbitsOfSymmetriesWithEssential_m(k)\to \OrbitsOfSymmetriesWithEssential_n(k)$ defined as
$$\SymmetricGroup_m\cdot M\mapsto \SymmetricGroup_n\cdot \Embedding_{m,n}(M) \text{ for } M\in\GL(m,2) \text{ with } \lvert \varepsilon(m) \rvert = k$$
The next lemma states that $ \Bijection_{m,n}^k|_{\OrbitsOfSymmetriesOnSphereWithEssential_m(d, k)}$
defines a bijection $\OrbitsOfSymmetriesOnSphereWithEssential_m(d, k)\to \OrbitsOfSymmetriesOnSphereWithEssential_n(d, k)$,
as long as $m\geq 2d$.

\begin{restatable}{lemma}{lemdistancefixedaftertwod}\label{lem:distance_fixed_after_two_d}
Fix some $d\in \Nats$, and let $k\leq 2d$.
Consider any matrix $M\in \GL(k,2)$ with $\Essential(M)=[k]$.
For any $n\geq 2d$, we have that $\Distance(\Embedding_{k,2d}(M)) = d$ if and only if $\Distance(\Embedding_{k,n}(M))=d$.
\end{restatable}
\begin{proof}
First, observe that the embedding $\Embedding_{k,n}$ can only decrease distances, i.e., for all $M\in \GL(k,2)$,
\[
\Distance(\Embedding_{k,n}(M))\leq \Distance(M).
\numberthis\label{eq:embedding_decreases_distance}
\]
This is because any circuit $\prod_{p=1}^d T^k_{i_p,j_p}$ evaluating to $M$ induces a circuit
\[
\Embedding_{k,n}(M)=\Embedding_{k,n}\left(\prod_{p=1}^dT_{i_p,j_p}\right)=\prod_{p=1}^d\Embedding_{k,n}(T_{i_p,j_p})
\] 
evaluating to $\Embedding_{k,n}(M)$.

Let now $N=\Embedding_{k,n}(M)$,
and assume that $\delta(N) \leq d$ so that there exists an optimal circuit $C$ evaluating to $N$, of length $c\leq d$.
By \cref{item:essential_invariant1} of \cref{lem:essential_indices_invariant}, $\Essential(N) = \Essential(M) = [k]$, and by \cref{item:essential_indeces_are_necessary} of \cref{lem:essential_indices_are_essential}, $[k]$ is a subset of the indices used by $C$. Since $C$ uses at most $2d$ distinct indices, this means that we can define a permutation $\sigma \in \SymmetricGroup_n$ with the property that $\sigma|_{[k]}$ is the identity and such that the image under $\sigma$ of the indices used by $C$ is a subset of $[2d]$.
Then, $\sigma \cdot N = N$, and $\sigma$ transforms $C$ into a circuit $C'$ that uses only indices from $[2d]$. In other words, every transvection in $C'$ lies in the image of the embedding $\Embedding_{2d,n}$. Applying the well-defined inverse $\Embedding_{2d,n}^{-1} : \mathrm{im}(\Embedding_{2d, n}) \to \GL(2d, 2)$ to each transvection in $C'$ gives a circuit that evaluates to $\Embedding_{k,2d}(M)$ and which has the same length as $C$.

We thus conclude that if $\Distance(\Embedding_{k, n}(M))\leq d$, then $\Distance(\Embedding_{k, 2d}(M))\leq \Distance(\Embedding_{k, n}(M))$.
This particularly holds when $\Distance(\Embedding_{k, n}(M))=d$,
and due to \cref{eq:embedding_decreases_distance}, we have $\Distance(\Embedding_{k, 2d}(M))=d$.

On the other hand, if $\Distance(\Embedding_{k, 2d}(M))=d$, \cref{eq:embedding_decreases_distance} implies that $\Distance(\Embedding_{k, n}(M))\leq d$.
Then the previous paragraph again concludes that $\Distance(\Embedding_{k, 2d}(M))\leq \Distance(\Embedding_{k, n}(M))$, hence $\Distance(\Embedding_{k, n}(M))= d$.
\qed
\end{proof}
Our final lemma relates the size of the orbits $U\in \OrbitsOfSymmetriesWithEssential_{m}(k)$ and $\Bijection_{m,n}^k(U)\in \OrbitsOfSymmetriesWithEssential_{n}(k)$.

\begin{restatable}{lemma}{lemsizeoforbitispolynomial}\label{lem:size_of_orbit_is_polynomial}
Let $M\in \GL(m,2)$ such that $\lvert\Essential(M)\rvert=k$. Then for any $n\geq m$, we have
\[
\lvert\SymmetricGroup_n\cdot \Embedding_{m,n}(M)\rvert=\lvert\SymmetricGroup_m\cdot M\rvert\cdot \binom{m}{k}^{-1} \binom{n}{k}.
\]
\end{restatable}
\begin{proof}
We first argue that it suffices to prove the statement for $m=k$.
Indeed, assuming it works for this case, by \cref{lem:canonical_representative}, we can find $M'\in \GL(k,2)$ such that $\Embedding_{k,m}(M')$ lies in the same orbit as $M$. 
Then, by applying \cref{lem:orbits_of_essential_indices_are_unique} to $M$ and $\Embedding_{k,m}(M')$, and by using the assumption that \cref{lem:size_of_orbit_is_polynomial} holds for $M'$, we get
\begin{align*}
\lvert\SymmetricGroup_n\cdot \Embedding_{m,n}(M)\rvert&=\lvert\SymmetricGroup_n\cdot \Embedding_{k,n}(M')\rvert \\ &=\lvert\SymmetricGroup_k\cdot M'\rvert\cdot\binom{n}{k} \\
&=\lvert\SymmetricGroup_m\cdot \Embedding_{k,m}(M')\rvert\cdot \binom{m}{k}^{-1}\cdot \binom{n}{k}\\
&=\lvert\SymmetricGroup_m\cdot M\rvert\cdot \binom{m}{k}^{-1}\cdot\binom{n}{k}.
\end{align*}

Suppose therefore $\Essential(M)=[m]$, and we will establish that
\[
\lvert\SymmetricGroup_n\cdot \Embedding_{m,n}(M)\rvert=\lvert\SymmetricGroup_m\cdot M\rvert\cdot \binom{n}{m}.
\]
By the orbit-stabilizer theorem (\cref{eq:orbit_stabilizer_size}), we have that 
\[
\lvert\SymmetricGroup_n\cdot \Embedding_{m,n}(M) \rvert = \frac{\lvert\SymmetricGroup_n\rvert}{\lvert\Stabilizer(\Embedding_{m,n}(M))\rvert} = \frac{n!}{\lvert\Stabilizer(\Embedding_{m,n}(M))\rvert}.
\]
We will now identify the stabilizer subgroup of $\Embedding_{m,n}(M)$, i.e., the set of permutations $\sigma\in\SymmetricGroup_n$ such that $\sigma\cdot \Embedding_{m,n}(M)=\Embedding_{m,n}(M)$.
Recall that $\Essential(M)=\Essential(\Embedding_{m,n}(M))$ (\cref{item:essential_invariant1} of \cref{lem:essential_indices_invariant}).
We consider two cases.

\SubParagraph{Case 1.}
Assume that $\sigma$ satisfies $\sigma(i)\in [m]$ if and only if $i\in [m]$, i.e., it only shuffles the essential indices of $\Embedding_{m,n}(M)$.
By \cref{item:essential_indeces_are_sufficient} of \cref{lem:essential_indices_are_essential}, there exists a circuit  $C=\prod_{k=1}^dT_{i_k,j_k}$ that evaluates to $\Embedding_{m,n}(M)$ and uses only the essential indices of $\Embedding_{m,n}(M)$.
By the definition of the group action, we have
\[
\sigma\cdot \left( \prod_{k=1}^dT_{i_k,j_k} \right) =  \prod_{k=1}^dT_{\sigma(i_k),\sigma(j_k)}.
\numberthis\label{eq:size_of_orbit_is_polynomial_circuit}
\]
By our assumption on $\sigma$, the circuit on the right-hand side of \cref{eq:size_of_orbit_is_polynomial_circuit} only uses the essential indices of $\Embedding_{m,n}(M)$.
Hence, we can take the preimage of $\Embedding_{m,n}$ to obtain a circuit evaluating to some element in $\GL(m,2)$.
It is therefore not hard to see that $\sigma$ is a stabilizer of $\Embedding_{m,n}(M)$ if and only if $\sigma'$ is a stabilizer of $M$,
where $\sigma'$ is obtained by ignoring the indices $m+1, \dots, n$ in $\sigma$.

\SubParagraph{Case 2.}
Assume that there is at least one essential index $i\in [m]$ such that $\sigma(i)\not \in [m]$, i.e., $\sigma$ maps an essential index of $\Embedding_{m,n}(M)$ to a non-essential index of $\Embedding_{m,n}(M)$.
Then $\sigma$ cannot be a stabilizer of $\Embedding_{m,n}(M)$, since we can again consider \cref{eq:size_of_orbit_is_polynomial_circuit}, and observe that at least one essential index of $M$ would be missing, meaning that $C$ cannot evaluate to $M$, due to \cref{item:essential_indeces_are_necessary} of \cref{lem:essential_indices_are_essential}.

Since the two cases are exhaustive, we conclude that  $\Stabilizer(\Embedding_{m,n}(M))$ is isomorphic to $\Stabilizer(M)\times \SymmetricGroup_{n-m}$ (i.e., only case 1 applies), to arrive at
\begin{align*}
\lvert\SymmetricGroup_n\cdot \Embedding_{m,n}(M)\rvert &= \frac{n!}{\lvert\Stabilizer(M)\times \SymmetricGroup_{n-m}\rvert} = \frac{n!}{\lvert \Stabilizer(M)\rvert \cdot (n-m)!} \\
&= \lvert\SymmetricGroup_m \cdot M|\frac{n!}{m!\cdot (n-m)!}
\end{align*}
where the last step is obtained by applying the orbit-stabilizer theorem (\cref{eq:orbit_stabilizer_size}) stating that $m!=\lvert\SymmetricGroup_m\rvert = \lvert\SymmetricGroup_m\cdot M\rvert \cdot \lvert\Stabilizer(M)\rvert$.
\qed
\end{proof}

\begin{restatable}{lemma}{witnessdist2ess2d}\label{lem:witness_dist_2_ess_2d}
    For any distance $d\geq 0$ the matrix
    $$M=T_{1,2}T_{3,4} \cdots T_{2d-1,2d}\in \GL(2d,2)$$
    has $2d$ essential indices and distance $d$.
\end{restatable}
\begin{proof}
    Evaluating the circuit that defines $M$ reveals $M[2i-1,2i]=1$ for every $i\in [d]$. Hence $\varepsilon(M)=[2d]$. By the definition of $M$, $\delta(M)\leq d$. On the other hand, by Item (1) of \cref{lem:essential_indices_are_essential} any circuit evaluating to $M$ must be of length at least $d$, meaning $\delta(M)\geq d$.
    \qed
\end{proof}

\Paragraph{Working with the full isometry group.}
Here we provide more details behind \cref{thm:sphere_polynomial_size_full_isometry}.
First, observe that the transpose-inverse map doesn't change essential indices. 
This is true, since on the generators this maps sends $T_{i,j}$ to $T_{j,i}$ so this is a consequence of \cref{item:essential_indeces_are_sufficient} of \cref{lem:essential_indices_are_essential}.

The set $\OrbitsOfSymmetriesWithEssential'(m)$ contains orbits of elements with $m$ essential indices, since \cref{lem:canonical_representative} still applies.
Hence, $\OrbitsOfSymmetriesOnSphereWithEssential'_n(d,m) = \OrbitsOfSymmetriesOnSphere_n(d)\cap \OrbitsOfSymmetriesWithEssential'_n(m)$ is a new set of orbits of the sphere at distance $d$ containing matrices with $m$ essential indices.
The following lemma is similar to \cref{lem:orbits_of_essential_indices_are_unique}, this time with respect to the group $\CyclicGroup_2$.

\begin{restatable}{lemma}{lemcyclicorbitsofmessentialindicesareunique}\label{lem:cyclic_orbits_of_m_essential_indices_are_unique}
Let $m\leq n$ and $M_1,M_2\in \GL(m,2)$. 
Then $M_1\in \CyclicGroup_2\cdot M_2$ if and only if $\Embedding_{m,n}(M_1)\in \CyclicGroup_2\cdot \Embedding_{m,n}(M_2)$.
\end{restatable}
\begin{proof}
It suffices to show the transpose-inverse map and $\Embedding_{m,n}$ commute. 
Transposing trivially commutes with $\Embedding_{m,n}$ and since $\Embedding_{m,n}$ is a group homomorphism, so does taking inverses.
\qed
\end{proof}

\thmspherepolynomialsizefullisometry*
\begin{proof}
Due to \cref{lem:cyclic_orbits_of_m_essential_indices_are_unique}, we obtain an induced bijection $$( \SymmetricGroup_m\times\CyclicGroup_2)\cdot M\mapsto ( \SymmetricGroup_n\times\CyclicGroup_2)\cdot \Embedding_{m,n}(M).$$
Consider an orbit $( \SymmetricGroup_n\times\CyclicGroup_2)\cdot M$ of some element $M\in \GL(n,2)$. 
By \cref{lem:commutativegroupaction} our actions commute which tells us that we can think of this as first finding the orbit from acting with $\CyclicGroup_2$ then compute the orbits $\SymmetricGroup_n\cdot M$ and $\SymmetricGroup_n\cdot (M^\top)^{-1}$, and take their union.

Let $M\in U$ be an arbitrary representative of an orbit in $\GL(n,2)/\SymmetricGroup_n$. 
If $-1\cdot M\in U$, then $\SymmetricGroup_n\cdot (-1\cdot M)=\SymmetricGroup_n\cdot M$, so the two orbits are identical, and thus acting by transpose-inverse does not add any new elements.
On the other hand if $-1\cdot M\not\in U$ (for all representatives $M$) then $\SymmetricGroup_n\cdot (-1\cdot M)\cap \SymmetricGroup_n\cdot M=\emptyset$ and $\lvert\SymmetricGroup_n\cdot (-1\cdot M)\rvert=\lvert\SymmetricGroup_n\cdot M\rvert$ since taking the transpose-inverse is an injective operation.
\\For an orbit $U\in \GL(n,2)/\SymmetricGroup_n$, define $\InverseTransposeOrbitCount(U)=1$ if $-1\cdot M\in U$ for some representative $M\in U$ and $\InverseTransposeOrbitCount(U)=2$ otherwise. 
The above observations show that
\begin{equation}\label{eq:kappa_eq1}\lvert( \SymmetricGroup_n\times\CyclicGroup_2)\cdot M \rvert = \InverseTransposeOrbitCount(U)\cdot \lvert U \rvert \text{ where } U=\SymmetricGroup_n\cdot M \end{equation}
for all $M \in \GL(n, 2)$.

Note also that for $m \leq n$ and any $M\in\GL(m,2)$, \cref{lem:orbits_of_essential_indices_are_unique} tells us that $-1\cdot M\in \SymmetricGroup_m\cdot M$ if and only if $\Embedding_{m,n}(-1\cdot M)\in \SymmetricGroup_n\cdot \Embedding_{m,n}(M)$, and in the proof of \cref{lem:cyclic_orbits_of_m_essential_indices_are_unique} we saw that $-1$ and $\phi_{m,n}$ commute, so we have
\begin{equation}\label{eq:kappa_eq2}
    \InverseTransposeOrbitCount(\SymmetricGroup_m\cdot M)=\InverseTransposeOrbitCount(\SymmetricGroup_n\cdot \Embedding_{m,n}(M)).
\end{equation}
Suppose that $\lvert\varepsilon(M)\rvert=k$.
Then
\begin{align*}
    \lvert( \SymmetricGroup_n\times\CyclicGroup_2)\cdot \phi_{m,n}(M)\rvert&=\InverseTransposeOrbitCount(\SymmetricGroup_n\cdot \phi_{m,n}(M))\cdot \lvert\SymmetricGroup_n\cdot \Embedding_{m,n}(M)\rvert\\
   &=\InverseTransposeOrbitCount(\SymmetricGroup_m\cdot M)\cdot \lvert\SymmetricGroup_m\cdot M\rvert\cdot \binom{m}{k}^{-1}\binom{n}{k}\\
    &=\lvert( \SymmetricGroup_m\times\CyclicGroup_2)\cdot M\rvert\cdot \binom{m}{k}^{-1}\binom{n}{k},
\end{align*}
where the first and final equalities follow from \cref{eq:kappa_eq1}, and the second equality follows from \cref{eq:kappa_eq2} and \cref{lem:size_of_orbit_is_polynomial}.
\qed
\end{proof}

\newpage
\section{Proofs of \cref{sec:permutations}}\label{sec:app_permutations}

\thmpermutationdistance*
\begin{proof}
First, recall \cref{lem:permutation_distance_upper_bound}, which states that $\Distance_n(\PermutationMatrix_{\sigma})\leq 3(n-\NumPermutationCycles(\sigma))$ for any permutation $\sigma\in\SymmetricGroup_n$.
Thus, \cref{conj:permutations} boils down to the lower bound
\[
\Distance_n(\PermutationMatrix_{\sigma}) \geq 3(n-\NumPermutationCycles(\sigma)).
\numberthis\label{eq:permutation_distance_conjecture_lower_bound}
\]

Take any $\sigma\in \SymmetricGroup_n$, let its cycle type be $(n_1,...,n_p)$, and assume that \cref{eq:permutation_distance_conjecture_lower_bound} is violated, i.e., $\PermutationMatrix_{\sigma}$ can be written as a product of $d$ transvections, with $d<3(n-p)$.
We will show there exists a long cycle whose distance is shorter than $3(n-1)$, thereby also violating \cref{eq:permutation_distance_conjecture_lower_bound}.

Indeed, note that two disjoint cycles $(a_1 a_2 ... a_k)$ and $(b_1 b_2 ... b_l)$ can be joined to form one cycle $(a_1 a_2 ... a_k b_1 b_2 ... b_l)$, by multiplying with the transposition $(a_1 b_1)$ from the left, i.e.,
\[
(a_1 b_1)(a_1 a_2 ... a_k)(b_1 b_2 ... b_l) = (a_1 a_2 ... a_k b_1 b_2 ... b_l).
\]
Hence, we can 'glue' together the $p$ disjoint cycles of $\sigma$ using $p-1$ transpositions, to form a long cycle $\tau$.
Since any transposition matrix $\PermutationMatrix_{(ij)}$ is the product of three transvections (\cref{item:transvection_identities5} of \cref{lem:transvection_identities}), $\PermutationMatrix_{\tau}$ can be written using 
\[
3(p-1)+d<3(p-1)+3(n-p)=3(n-1)
\]
transvections.
Thus $\Distance(\PermutationMatrix_{\tau}) < 3(n-1)$, as desired.
\qed
\end{proof}
\newpage
\section{Experimental Details}
\label{app:tables}

Our experiments were run on two machines, a fast 128-core machine, to speed up computations, and a slower 40-core machine
with 1.5TB for cases where memory usage was the bottleneck.
The specifications of these machines are:
\begin{itemize}
    \item A fast 128-core machine with 768GB of internal memory; each core running at a frequency of 3.1GHz (AMD EPYC 9554).
    \item A large 40-core machine with 1.5TB of internal memory; each core running at a frequency of 2.1GHz (Intel Xeon Gold 6230).
\end{itemize}

As stated in the main paper, we conduct a BFS search in
$G_n/{\SymmetricGroup_n}$, storing one representative per orbit.
For $n=1,\ldots,8$, matrices can be stored in a single 64-bit word.
We enumerate the representatives in each BFS level, generate their successors, compute their representatives, and test those against the previous and the current BFS level. If they are new, we store them in the next BFS level.

We compute unique representatives as canonical isomorphic graphs, using the \texttt{nauty}-software~\cite{McKay2014}.
We also keep a global count of the sizes of all orbits that we encounter (derived from the number of automorphisms reported by \texttt{nauty}).
We modified \texttt{nauty}'s code to count in 64-bit integers rather than double floats, in order to avoid approximation errors, while testing that we didn't overflow.

To achieve parallel speedup, we store all representatives of a level in a lock-free concurrent hash table, following the design in \cite{DBLP:conf/fmcad/LaarmanPW10} and using an implementation from \cite{DBLP:conf/nfm/Berg21}. The elements of the level are enumerated and processed in parallel (i.e., each worker takes some batches from the current BFS level), relying on OpenMP.

\subsection*{Cayley Graph for $n=1,\ldots,7$}

We could enumerate the full quotient graph for $n=1,\ldots,7$.
We report the size of each sphere, $\lvert\Sphere_n(d)\rvert$
(\cref{tab:levelsizes}) and the size of the stored levels,
$\lvert\Sphere_n(d)/{\SymmetricGroup_n}\rvert$ (\cref{tab:orbits}) for all
relevant distances $d$. As a sanity check on the implementation, \cref{tab:levelsizes}
also shows that for $n=1,..,7$, the sum of the sphere sizes corresponds
with the size of the group $\GL(n,2)$.

For $N=6$, the computation took merely 3s on the 128-core machine.
The full computation for $N=7$ took 8483s (less than 2.5 hours) on the 128-core machine.
The largest level contains 
13,616,116,190 orbits ($d=14$, \cref{tab:orbits}), stored in a concurrent hash-table of size $2^{35}$ nodes.

\subsection*{Cayley Graph for $n=8$ (up to $d=12$)}
For $n=8$, we could compute all BFS levels up to $d=12$.
The corresponding sphere
$\Sphere_8(12)$ contains 1,342,012,729,372,308 elements (rightmost column in \cref{tab:levelsizes}), partitioned in 
33,719,514,377 orbits (rightmost column in \cref{tab:orbits}). We needed a machine 
with 1.5TB internal memory to store this large BFS level. This computation took 1.5 hours on 40 cores.

We also conducted a bi-directional search, between the identity matrix $\IdentityMatrix_8$ and the long permutation cycle 
$\lambda=(1\, 2\, 3\, 4\, 5\, 6\, 7\, 8)$. This search terminated at $F_{12}$, after 12 forward steps from $\IdentityMatrix_8$ and at $B_{9}$, 9 backward steps from $\PermutationMatrix_\lambda$. So their distance is indeed 21 steps.

$F_{12}$ contains 33,719,514,377 orbits (as above), and $B_9$ contains 65,936,050,032 orbits, stored in a concurrent hash-table of $2^{36}$ nodes.
This bidirectional search took 5.5 hours on the 1.5TB/40-core machine.

\subsection*{Computations for $n=20$ (up to $d=10$)}
We computed 10 levels for $n=20$, in order to compute the coefficients of the 20-degree polynomial $f_{10}(n)$.
The sphere at $d=10$ contains 743,188,850 orbits. Since $n=20$ has a large symmetry group, the orbits themselves can be very large, representing in total 16,798,138,692,326,241,596 matrices ($\approx 16.8\times 10^{18}$).
This computation took 2286s (less than 40 minutes) on the 40 core machine. The computed coefficients of $f_{10}(n)$ are reported in the rightmost column of \cref{tab:polynomials}.

\begin{table}[ht] 
\caption{Size of spheres, $|\Sphere_n(d)|$ for $n=1,\ldots,7$:
How many different \cnot \ circuits on $n$
qubits require exactly $d$ \cnot\ gates ($d=0$ corresponds to
the empty circuit, for $I_n$).
We also include the partial data for $n=8$ up to $d=12$.
At the bottom, we compare the sum of the sphere sizes to
the size of the group $\GL(n,2)$.
\medskip
\label{tab:levelsizes}}
\centering
\scalebox{0.77}{
\begin{tabular}{|r|ccccccc| c|}
\hline
  & \multicolumn{7}{|c|}{$n$} & \\
$d$ & 1 & 2 & 3 & 4 & 5 & 6 & 7 & 8\\
\hline
0 & 1 & 1 & 1 & 1 & 1 & 1 & 1 & 1 \\
1 & - & 2 & 6 & 12 & 20 & 30 & 42 & 56 \\
2 & - & 2 & 24 & 96 & 260 & 570 & 1092 & 1904 \\
3 & - & 1 & 51 & 542 & 2570 & 8415 & 22141 & 50316 \\
4 & - & - & 60 & 2058 & 19680 & 101610 & 375480 & 1121820 \\
5 & - & - & 24 & 5316 & 117860 & 1026852 & 5499144 & 21927640 \\
6 & - & - & 2 & 7530 & 540470 & 8747890 & 70723842 & 383911500 \\
7 & - & - & - & 4058 & 1769710 & 61978340 & 801887394 & 6086458100 \\
8 & - & - & - & 541 & 3571175 & 355193925 & 7978685841 & 87721874450 \\
9 & - & - & - & 6 & 3225310 & 1561232840 & 68818316840 & 1148418500236 \\
10 & - & - & - & - & 736540 & 4753747050 & 503447045094 & 13587845739286 \\
11 & - & - & - & - & 15740 & 8111988473 & 3008371364033 & 143890218187240 \\
12 & - & - & - & - & 24 & 4866461728 & 13735773412074 & 1342012729372308 \\
13 & - & - & - & - & - & 437272014 & 42362971639322 & ??? \\
14 & - & - & - & - & - & 949902 & 68493002803224 & ??? \\
15 & - & - & - & - & - & 120 & 33871696277888 & ??? \\
16 & - & - & - & - & - & - & 1796520274568 & ??? \\
17 & - & - & - & - & - & - & 534600540 & ??? \\
18 & - & - & - & - & - & - & 720 & ??? \\
\hline
Sum & 1 & 6 & 168 & 20160 & 9999360 & 20158709760 & 163849992929280 & (1500733427144857)\\
$|\GL(n, 2)|$ & 1 & 6 & 168 & 20160 & 9999360 & 20158709760 & 163849992929280 & 5348063769211699200\\
\hline
\end{tabular}
}
\end{table}

\begin{table}[ht]
\caption{$\lvert\Sphere_n(d)/{\SymmetricGroup_n}\rvert$: The number of orbits for $n=1,\ldots,7$ at level $d$. This corresponds to the size of the physically stored BFS levels. We also include the partial data for $n=8$ up to $d=12$. 
The bold entries at $n=2d$ indicate where the number of orbits becomes stable for a particular distance $d$.
\medskip
\label{tab:orbits}}
\centering
\scalebox{1}{
\begin{tabular}{|r|ccccccc|c|}
\hline
  & \multicolumn{7}{|c|}{$n$} & \\
$d$ & 1 & 2 & 3 & 4 & 5 & 6 & 7 & 8\\
\hline
0 & 1 & 1 & 1 & 1 & 1 & 1 & 1 & 1 \\
1 & - & {\bf 1} & 1 & 1 & 1 & 1 & 1 & 1 \\
2 & - & 1 & 5 & {\bf 6} & 6 & 6 & 6 & 6 \\
3 & - & 1 & 9 & 27 & 31 & {\bf 32} & 32 & 32 \\
4 & - & - & 12 & 94 & 200 & 228 & 232 & {\bf 233} \\
5 & - & - & 4 & 238 & 1069 & 1767 & 1941 & 1969 \\
6 & - & - & 1 & 334 & 4740 & 13425 & 18618 & 19855 \\
7 & - & - & - & 181 & 15198 & 90507 & 181632 & 223299 \\
8 & - & - & - & 25 & 30461 & 506752 & 1687466 & 2653755 \\
9 & - & - & - & 1 & 27333 & 2202850 & 14102906 & 31414389 \\
10 & - & - & - & - & 6236 & 6672137 & 101627779 & 353662338 \\
11 & - & - & - & - & 134 & 11342151 & 602662335 & 3657182348 \\
12 & - & - & - & - & 1 & 6786712 & 2741492657 & 33719514377 \\
13 & - & - & - & - & - & 609993 & 8436220042 & ??? \\
14 & - & - & - & - & - & 1359 & 13616116190 & ??? \\
15 & - & - & - & - & - & 1 & 6726326530 & ??? \\
16 & - & - & - & - & - & - & 356621214 & ??? \\
17 & - & - & - & - & - & - & 106744 & ??? \\
18 & - & - & - & - & - & - & 1 & ??? \\
\hline
Sum & 1 & 4 & 33 & 908 & 85411 & 28227922 & 32597166327 & (37764672603)\\
\hline
\end{tabular}
}
\end{table}

\end{document}